\documentclass[11pt]{article}
\usepackage{graphicx}
\usepackage{lscape}
\usepackage{amsmath}
\usepackage{amsthm}
\usepackage{graphicx}
\usepackage{amsfonts}
\usepackage{amssymb}
\usepackage{subfigure}
\usepackage{natbib}

\newcommand{\EE}{\mathbb{E} }

\newcommand{\PP}{\mathbb{P} }

\usepackage{color}

\newcommand{\be}{\begin{equation}}
\newcommand{\en}{\end{equation}}

\newtheorem{theorem}{Theorem}
\newtheorem{prop}{Proposition}

\newtheorem{remark}{Remark}

\newcommand{\ea}{\end{eqnarray}}
\newcommand{\ba}{\begin{eqnarray}}
\newcommand{\ean}{\end{eqnarray*}}
\newcommand{\ban}{\begin{eqnarray*}}

\graphicspath{{hetergroup/}}

\usepackage[margin=1in]{geometry}

\begin{document}
\title{Systemic Risk and Heterogeneous Mean Field Type Interbank Network}
\author{Li-Hsien Sun\thanks{Institute of Statistics, National Central University, Chung-Li, Taiwan, 32001 {\em lihsiensun@ncu.edu.tw}. Work supported by Most grant 108-2118-M-008-002-MY2}}

\date{}
\maketitle
\begin{abstract}
We study the system of heterogeneous interbank lending and borrowing based on the relative average of log-capitalization given by the linear combination of the average within groups and the ensemble average and describe the evolution of log-capitalization by a system of coupled diffusions. The model incorporates a game feature with homogeneity within groups and heterogeneity between groups where banks search for the optimal lending or borrowing strategies through minimizing the heterogeneous linear quadratic costs in order to avoid to approach the default barrier.
Due to the complicity of the lending and borrowing system, the closed-loop Nash equilibria and the open-loop Nash equilibria are both driven by the coupled Riccati equations. The existence of the equilibria in the two-group case where
the number of banks are  sufficiently large is guaranteed by the solvability for the coupled Riccati equations as the number of banks goes to infinity in each group. The equilibria are consisted of the mean-reverting term identical to the one group game and the group average owing to heterogeneity. In addition, the corresponding heterogeneous mean filed game with the arbitrary number of groups is also discussed. The existence of the $\epsilon$-Nash equilibrium in the general $d$ heterogeneous groups is also verified. Finally, in the financial implication, we observe the Nash equilibria governed by the mean-reverting term and the linear combination of the ensemble averages of individual groups and study the influence of the relative parameters on the liquidity rate through the numerical analysis.

\end{abstract}
\textbf{Keywords:} Systemic risk, inter-bank borrowing and lending system, heterogeneous group, relative ensemble average, Nash equilibrium, Mean Field Game.


\section{Introduction}\label{sec:intro}
 Toward a deeper understanding of systemic risk created by lending and borrowing behavior under the heterogeneous environment, we extend the model studied in \cite{R.Carmona2013} from one homogeneous group into several groups with heterogeneity.
The evolution of monetary reserve is described by a system of interacting diffusions with homogeneity within groups and heterogeneity between groups. Banks intend to borrow money from a central bank when they remain below the global average of capitalization treated as the critical level and lend money to a central bank when they stay above the same critical level through minimizing their cost for the corresponding lending or borrowing given by the linear quadratic objective functions with varied parameters. Furthermore, motivated by \cite{Touzi2015}, instead of the global ensemble average, we propose the relative ensemble average which is the linear combination of the group average and the global ensemble average. In the case of the finite players with heterogeneous groups, we solve the closed-loop Nash equilibria using the dynamic programming principle and the corresponding fully coupled backward Hamilton Jacobi Bellman (HJB) equations. In addition, through the Pontryagin minimization principle and the corresponding the adjoint forward and backward stochastic differential equations (FBSDEs), we get the open-loop Nash equilibria. 
Due to the complicity of the heterogeneity, the closed-loop Nash equilibria and the open-loop Nash equilibria are given by the coupled Riccati equations. Hence, in the two-group case, we propose the solvability condition for the existence of the equilibria as the number of banks in each group goes to infinity in the sense that in the case of sufficiently large number of banks, the existence of the closed-loop Nash equilibria and the open-loop Nash equilibria can be guaranteed. Furthermore, we discuss heterogeneous mean field games (MFGs) with the common noises where the number of groups is arbitrary. Due to the complexity generated by common noises, the $\epsilon$-Nash equilibria can not be obtained using the HJB equations. Hence, the adjoint FBSDEs are applied to solve the equilibria. The existence of the $\epsilon$-Nash equilibria is also proved under some sufficient conditions. As the results in \cite{R.Carmona2013},  the closed-loop Nash equilibria and the open look Nash equilibria converge to the $\epsilon$-Nash equilibria. 
We observe that owing to the linear quadratic regulator, the equilibria are the linear combination of the mean-reverting term identical to the one group system discussed in \cite{R.Carmona2013} and the group ensemble averages given by heterogeneity. In addition, through the numerical studies, if banks prefer tracing the global ensemble average rather than the average of their own groups, they intend to increase the liquidity rate, and the larger sample size also implies the larger liquidity rate.

In the literature, this type of interaction in continuous time is studied in several models.  \cite{Fouque-Sun, R.Carmona2013} describe systemic risk based on the coupled Ornstein-Uhlenbeck (OU) type processes. The extension of this OU type model with delay obligation is also proposed by \cite{Carmona-Fouque2016}. \cite{Fouque-Ichiba, Sun2016} investigate system crash through the Cox-Ingersoll-Ross (CIR) type processes.  The rare event related to systemic risk given by the bistable model is discussed in \cite{Garnier-Mean-Field, GarnierPapanicolaouYang}. The stability created by a central agent according to the model in \cite{Garnier-Mean-Field, GarnierPapanicolaouYang} is provided by \cite{Papanicolaou-Yang2015}.

The asymptotic equilibria called $\epsilon$-Nash equilibria of stochastic differential games in one homogeneous group are solved by MFGs proposed by \cite{MFG1, MFG2, MFG3}. Here, the interactions given by empirical distributions whose solution satisfies the coupled backward HJB equation backward in time and the forward Kolmogorov equation forward in time. Independently, Nash certainty equivalent treated as a similar case of MFGs is also developed by \cite{HuangCainesMalhame1,HuangCainesMalhame2}. In addition, the probabilistic approach in the form of FBSDEs to obtain $\epsilon$-Nash equilibria is studied in
\cite{Bensoussan_et_al, CarmonaDelarueLachapelle, Carmona-Lacker2015,MFG-book-1}. \cite{Lacker-Webster2014, Lacker2015,  MFG-book-2} discuss MFGs with common noise and the master equations.  In particular, \cite{Lacker-Zari2017} study the optimal investment under heterogeneous relative performance in the case of mean field limit.

The paper is organized as follows. In Section \ref{Heter}, we analyze the case of two heterogeneous groups with relative performance and study the corresponding closed-loop Nash equilibria in Section \ref{Sec:NE}. In parciular, Section \ref{MFG} is devoted to solving for the $\epsilon$-Nash equilibria with common noises in MFGs with heterogeneous groups using the coupled FBSDEs. The financial implication is illustrated in Section \ref{FI}. The concluding remark is given in Section \ref{conclusions}.

\section{Heterogeneous Groups}\label{Heter}
According to the interbank lending and borrowing system discussed in \cite{R.Carmona2013}, it is nature to consider the model of interbank lending and borrowing containing $d$ groups and $N_k$ denoted the number of banks in group $k=1,\cdots,d$ where $N=\sum_{j=1}^dN_j$. The $i$-th Bank in group $k$ intends to obtain the optimal strategy through minimizing its own linear quadratic cost function
\ba
\label{objective}J^{(k)i}(\alpha)=\EE\bigg\{\int_{0}^{T}  f^{N}_{(k)}(X^{(k)i}_t, X^{-(k)i}_t, \alpha^{(k)i}_t)
dt+  g_{(k)}(X^{(k)i}_T, X^{-(k)i}_T )\bigg\},
\ea
where $X=(X^{(1)1},\cdots,X^{(d)N_d})$, $x=(x^{(1)1},\cdots,x^{(d)N_d})$, $\alpha=(\alpha^{(1)1},\cdots,\alpha^{(d)N_d})$, and $ X^{-{(k)i}}=( X^{(1)1},\cdots, X^{(k)i-1},  X^{(k)i+1},\cdots, X^{(d)N_d})$ where the running cost is
\ba
  &&f^N_{(k)}(x^{(k)i}, x^{-(k)i},\alpha)=\frac{(\alpha )^2}{2}-q_k\alpha \left(\tilde x^{\lambda_k} -x^{(k)i}\right)
+\frac{\epsilon_k}{2}\left(\tilde x^{\lambda_k} -x^{(k)i}\right)^2,\label{running_cost}
\ea
and terminal cost is
\be\label{teminal_cost}
g^N_{(k)}(x^{(k)i}, x^{-(k)i})=\frac{c_k}{2}\left(\tilde x^{\lambda_k}  -x^{(k)i} \right)^2,
\en
with $x^{-(k)i}=( x^{(1)1},\cdots, x^{(k)i-1},  x^{(k)i+1},\cdots, x^{(d)N_d}$) where the relative ensemble average is 
\[
\overline x^{\lambda_k}=(1-\lambda_k)\overline x^{(k)}+\lambda_k\overline{x}
\]
where the global average of capitalization and the group average of capitalization are written as 
\[
\overline{x} =\frac{1}{N}\sum_{k=1}^d\sum_{i=1}^{N_k}x^{(k)i} ,\,\overline{x} ^{(k)}=\frac{1}{N_k}\sum_{i=1}^{N_k}x^{(k)i},
\]
under the constraint
\ba\label{diffusions}
\nonumber dX^{(k)i}_t &=& (\alpha_{t}^{(k)i}+\gamma^{(k)}_{t})dt\\
&&+\sigma_k\left(\rho dW^{(0)}_t+\sqrt{1-\rho^2}\left(\rho_{k}dW_t^{(k)}+\sqrt{1-\rho^2_{k}}dW^{(k)i}_t\right)\right),
\ea
for $i=1,\cdots,N_k$ with nonnegative diffusion parameters $\sigma_k$, nonnegative deterministic growth rate $\gamma^{(k)}$ in $L^\infty$-space denoted as the collection of bounded measurable functions on $[0,T]$. We further assume that $W^{(k)i}_t$ for all $k=1,\cdots,d$ and $i=1,\cdots,N_k$ are standard Brownian motions and $W^{(0)}_t$ and $W^{(k)}_t$ for $k=1,\cdots,d$ are the common noises between groups and within groups corresponding to the parameters $\rho $ and $\rho_k$ for $k=1,\cdots,d$ denoted as the correlation between groups and within groups respectively. Note that all Brownian motions are defined on a filtration probability space $(\Omega,{\cal{F}},{\cal{P}},\{{\cal{F}}_t\})$ referred to as Definition 5.8 of Chapter 2 in \cite{Karatzas2000}. The initial value $X_0^{(k)i}$  which may also be a squared integrable random variable $\xi^{(k)}$ for $k=1,\cdots,d$ independent of the Brownian motions defined on $(\Omega,{\cal{F}},{\cal{P}},\{{\cal{F}}_t\})$. Namely, the system of the interbank lending and borrowing contains the feature of homogeneity within the groups and heterogeneity between groups. Note that $\alpha_\cdot$ is a progressively measurable control process and  $\alpha^{(k)i}_\cdot$ is admissible if it satisfy the integrability condition given by
\be\label{admissible}
\EE\int_0^T\left|\alpha_s^{(k)i}\right|^2ds<\infty.
\en


In addition, the parameters $q_k$, $\epsilon_k$, and $c_k$ stay in positive satisfying $q_k^2\leq \epsilon_k$ in order to guarantee  that $\alpha\rightarrow f_{(k)}^{N}(x,\alpha)$ is convex for any $x$ and $x\rightarrow f_{(k)}^{N}(x,\alpha)$ is convex for any $\alpha$.  
The parameters $0\leq\lambda_k\leq 1$ for $k=1,\cdots,d$ are described as the relative consideration for the group average and the global average. The case of large $\lambda$ means that banks consider tracing the global average rather than the group one through the large ratio on global average.  Finally, $q_k$ presents the incentive of lending and borrowing behavior for banks in group $k$ as the description in \cite{Carmona-Fouque2016,R.Carmona2013}.  

For simplicity, in case of the finite players, we study the case of two heterogeneous groups where $d=2$ 
Hence, the dynamics for both groups are written as 
\ba
\nonumber dX^{(1)i}_t &=& (\alpha_{t}^{(1)i}+\gamma^{(1)}_{t})dt\\
\label{diffusion-major}&&+\sigma_1\left(\rho dW^{(0)}_t+\sqrt{1-\rho^2}\left(\rho_{1}dW_t^{(1)}+\sqrt{1-\rho^2_{1}}dW^{(1)i}_t\right)\right),
\ea
and 
\ba
\nonumber dX^{(2)i}_t &=& (\alpha_{t}^{(2)i}+\gamma^{(2)}_{t})dt\\
\label{diffusion-minor}&&+\sigma_2\left(\rho dW^{(0)}_t+\sqrt{1-\rho^2}\left(\rho_{2}dW_t^{(2)}+\sqrt{1-\rho^2_{2}}dW^{(2)i}_t\right)\right),
\ea
with the initial value $X_0^{(k)i}$ for $k=1,2$ and $i=1,\cdots,N_k$. In particular, given the first group consisted of larger banks and the second group consisted of smaller banks, we may further assume that $ 0\leq\lambda_1<\lambda_2 \leq 1$ since large banks intend to trace their own group ensemble average $\overline X^{(1)}$ rather than the global average $\overline X$. On the contrary, small banks prefer tracing large banks through the global ensemble average. In addition, the number of large banks $N_1$ is usually less than the number of small banks $N_2$.

\section{Construction of Nash Equilibria }\label{Sec:NE}
This section is devoted to obtain the closed-loop Nash equilibria and the open-loop Nash equilibria in the finite players game. The solution to the closed-loop Nash equilibria are given by the HJB approach. The open-loop Nash equilibria are obtained using the FBSDEs based on the Pontryagin minimum principle.

\subsection{Closed-loop Nash Equilibria}\label{HJB-approach}
In order to solve the closed-loop Nash equilibrium,  given the optimal strategies $\hat\alpha^{(k)j}$ for $j\neq i$ with the corresponding trajectories $$\hat X^{-{(k)i}}=(\hat X^{(1)1},\cdots,\hat X^{(k)i-1},\hat X^{(k)i+1},\cdots,\hat X^{(2)N_2}),$$  bank $(1)i$ and bank $(2)j$ intend to minimize the objective functions through the value functions written as
\be\label{value-function-1}
V^{(1)i}(t,x)=\inf_{\alpha^{(1)i}\in{\cal{A}} }\EE_{t,x}\left\{\int_t^T  f^N_{(1)}(X^{(1)i}_s, \hat X^{-(1)i}_s, \alpha^{(1)i}_s )ds+ g^N_{(1)}(X_{T}^{(1)i},\hat X^{-(1)i}_T )\right\},
\en
and 
\be\label{value-function-2}
V^{(2)j}(t,x)=\inf_{\alpha^{(2)j}\in{\cal{A}}}\EE_{t,x}\left\{\int_t^T  f^N_{(2)}(X^{(2)j}_s, \hat X^{-(2)j}_s, \alpha^{(2)j}_s )ds+ g^N_{(2)}(X_{T}^{(2)j},\hat X^{-(2)j}_T )\right\},
\en
subject to 
\ba\label{coupled-1}
\nonumber dX^{(1)i}_t &=& (\alpha_{t}^{(1)i}+\gamma^{(1)}_{t})dt\\
&&+\sigma_1\left(\rho dW^{(0)}_t+\sqrt{1-\rho^2}\left(\rho_{1}dW_t^{(1)}+\sqrt{1-\rho^2_{1}}dW^{(1)i}_t\right)\right) ,
\ea
and 
\ba\label{coupled-2}
\nonumber dX^{(2)j}_t &=& (\alpha_{t}^{(2)j}+\gamma^{(2)}_{t})dt\\
&&+\sigma_2\left(\rho dW^{(0)}_t+\sqrt{1-\rho^2}\left(\rho_{1}dW_t^{(2)}+\sqrt{1-\rho^2_{2}}dW^{(2)j}_t\right)\right),
\ea
where $W^{(k)i}_t$ for all $k=1,2$, $i=1,\cdots,N_1$ and $j=1,\cdots,N_2$ are standard Brownian motions and $W^{(0)}_t$ and $W^{(k)}_t$ for $k=1,2$ are the common noises between groups and within groups corresponding to the parameters $\rho$ and $\rho_k$ for $k=1,\cdots,2$ denoted as the correlation between groups and within groups respectively. The initial value $X_0^{(k)i}$ may also be a squared integrable random variable $\xi^{(k)}$. The control process $\alpha^{(k)i}$ is progressively measurable and ${\cal{A}}$ is denoted as the admissible set given by \eqref{admissible} for $\alpha^{(k)i}$. 
Note that $\EE_{t,x}$ denotes the expectation given $X_t=x$. 

\vskip 0.5 in

\begin{theorem}\label{Hete-Nash}
Assuming $q_k^2\leq \epsilon_k$ for $k=1,2$ and $\eta^{(i)}_t$ and $\phi^{(i)}_t$ for $i=1,\cdots,10$ satisfying \eqref{eta1} to \eqref{phi10}, the value functions of the closed-loop Nash equilibria to the problem \eqref{value-function-1} and \eqref{value-function-2} subject to \eqref{coupled-1} and \eqref{coupled-2} are given by 
\ba
\nonumber V^{(1)i}(t,x)&=&\frac{\eta^{(1)}_t}{2}( \overline x^{(1)}-x^{(1)i})^2+\frac{\eta^{(2)}_t}{2}(\overline x^{(1)})^2+\frac{\eta^{(3)}_t}{2}(\overline x^{(2)})^2\\
\nonumber&&+\eta^{(4)}_t( \overline x^{(1)}-x^{(1)i})\overline x^{(1)}+\eta^{(5)}_t(\overline x^{(1)}-x^{(1)i})\overline x^{(2)}+\eta^{(6)}_t\overline x^{(1)}\overline x^{(2)}\\
&&+\eta^{(7)}_t( \overline x^{(1)}-x^{(1)i})+\eta^{(8)}_t\overline x^{(1)}+\eta^{(9)}_t\overline x^{(2)}+\eta^{(10)}_t,\label{ansatz-1-prop}
\ea
and 
\ba
\nonumber V^{(2)j}(t,x)&=&\frac{\phi^{(1)}_t}{2}( \overline x^{(2)}-x^{(2)j})^2+\frac{\phi^{(2)}_t}{2}(\overline x^{(1)})^2+\frac{\phi^{(3)}_t}{2}(\overline x^{(2)})^2\\
\nonumber&&+\phi^{(4)}_t(\overline x^{(2)}-x^{(2)j})\overline x^{(1)}+\phi^{(5)}_t(\overline x^{(2)}-x^{(2)j})\overline x^{(2)}+\phi^{(6)}_t\overline x^{(1)}\overline x^{(2)}\\
&&+\phi^{(7)}_t(\overline x^{(2)}-x^{(2)j})+\phi^{(8)}_t\overline x^{(1)}+\phi^{(9)}_t\overline x^{(2)}+\phi^{(10)}_t,\label{ansatz-2-prop}
\ea
and the corresponding closed-loop Nash equilibria are  
\ba
\label{optimal-finite-ansatz-V1}
&&\hat\alpha^{(1)i}(t,x)=(q_1+\widetilde\eta^{(1)}_t)( \overline x^{(1)}-x^{(1)i})+\widetilde\eta^{(4)}_t\overline x^{(1)}+\widetilde\eta^{(5)}_t\overline x^{(2)}+\widetilde\eta_t^{(7)},\\
&&\hat\alpha^{(2)j}(t,x)=(q_2+\widetilde\phi^{(1)}_t)( \overline x^{(2)}-x^{(2)j})+\widetilde\phi^{(4)}_t\overline x^{(1)}+\widetilde\phi^{(5)}_t\overline x^{(2)}+\widetilde\phi_t^{(7)},\label{optimal-finite-ansatz-V2}
\ea
for $i=1,\cdots,N_1$ and $j=1,\cdots,N_2$ where 
\ba
&&\nonumber \widetilde\eta^{(1)}_t=(1-\frac{1}{N_1})\eta^{(1)}_t-\frac{1}{N_1}\eta^{(4)}_t, \quad\widetilde\eta^{(4)}_t=(1-\frac{1}{N_1})\eta^{(4)}_t-\frac{1}{N_1}\eta^{(2)}_t-\lambda_1(1-\beta_1)q_1,\\
&&\label{tildeeta}\widetilde\eta^{(5)}_t=(1-\frac{1}{N_1})\eta^{(5)}_t-\frac{1}{N_1}\eta^{(6)}_t+\lambda_1\beta_2q_1,\quad\widetilde\eta^{(7)}_t=(1-\frac{1}{N_1})\eta^{(7)}_t-\frac{1}{N_1}\eta^{(8)}_t,
\ea
and 
\ba
&&\nonumber \widetilde\phi^{(1)}_t=(1-\frac{1}{N_2})\phi^{(1)}_t-\frac{1}{N_2}\phi^{(5)}_t, \quad \widetilde\phi^{(4)}_t=(1-\frac{1}{N_2})\phi^{(4)}_t-\frac{1}{N_2}\phi^{(6)}_t+\lambda_2\beta_1q2,\\
\nonumber&&\widetilde\phi^{(5)}_t=(1-\frac{1}{N_2})\phi^{(5)}_t-\frac{1}{N_2}\phi^{(3)}_t-\lambda_2(1-\beta_2)q_2,\quad \widetilde\phi^{(7)}_t=(1-\frac{1}{N_2})\phi^{(7)}_t-\frac{1}{N_2}\phi^{(9)}_t.\\
\label{tildephi}
\ea
\end{theorem}
\begin{proof}
See Appendix \ref{Appex-1}. 
\end{proof}

Similarly, due to the complexity of the coupled ODEs given by (\ref{eta1}-\ref{phi10}), we now study the existence of the coupled system (\ref{eta1}-\ref{phi10}) in the case of $N_1\rightarrow\infty$ and $N_2\rightarrow\infty$.

\begin{prop}\label{Prop_suff}
As $N_1\rightarrow\infty$ and $N_2\rightarrow\infty$, the coupled equations \eqref{eta1} to \eqref{eta6} and \eqref{phi1} to \eqref{phi6} are rewritten as 
written as 
\ba
\label{eta1_N} \dot{\widehat{\eta}}^{(1)}_t&=&2 q_1\widehat\eta^{(1)}_t+ (\widehat\eta_t^{(1)})^2-(\epsilon_1-q_1^2),\\
 \nonumber  {\dot{\widehat\eta}^{(2)}_t}&=&2\left(-\widehat\eta^{(4)}_t+q_1\lambda_1(1-\beta_1)\right) \widehat\eta^{(2)}_t-(\widehat\eta_t^{(4)})^2\\
&&-2\left(\widehat\phi^{(4)}_t+q_2\lambda_2\beta_1\right)\widehat\eta^{(6)}_t-(\epsilon_1-q_1^2)\lambda_1^2(\beta_1-1)^2,\label{eta2_N}\\
\nonumber   {\dot{\widehat\eta}^{(3)}_t} &=&2\left(-\widehat\phi^{(5)}_t+q_2\lambda_2(1-\beta_2)\right)\widehat\eta^{(3)}_t-(\widehat\eta_t^{(5)})^2\\
&&-2\left(\widehat\eta^{(5)}_t+q_1\lambda_1\beta_2\right)\widehat\eta^{(6)}_t-(\epsilon_1-q_1^2)\lambda_1^2\beta_2^2\label{eta3_N}\\
\nonumber \dot{\widehat\eta}^{(4)}_t&=&q_1\left(1+\lambda_1(1-\beta_1)\right) \widehat\eta^{(4)}_t-(\widehat\eta^{(4)}_t)^2\\
 &&-\left(\widehat\phi^{(4)}_t+q_2\lambda_2\beta_1\right)\widehat\eta^{(5)}_t+(\epsilon_1-q_1^2)\lambda_1(1- \beta_1),\label{eta4_N}\\
 \dot{\widehat\eta}^{(5)}_t&=&\left(q_1-\widehat\eta^{(4)}_t-\widehat\phi^{(5)}_t+q_2\lambda_2(1-\beta_2)\right) \widehat\eta^{(5)}_t-q_1\lambda_1\beta_2\widehat\eta^{(4)}_t-(\epsilon_1-q_1^2)\lambda_1\beta_2,\label{eta5_N}\\ 
\nonumber\dot{\widehat\eta}^{(6)}_t&=&\left(-\widehat\eta^{(4)}_t-\widehat\phi^{(5)}_t+q_1\lambda_1(1-\beta_1)+q_2\lambda_2(1-\beta_2)\right)\widehat\eta^{(6)}_t-\left(\widehat\eta^{(5)}_t+q_1\lambda_1\beta_2\right)\widehat\eta^{(2)}_t\\
 &&-\widehat\eta_t^{(4)}\widehat\eta_t^{(5)}-\left(\widehat\phi^{(4)}_t+q_2\lambda_2\beta_1\right)\widehat \eta^{(3)}_t+(\epsilon_1-q_1^2)\lambda_1^2(1-\beta_1)\beta_2,\label{eta6_N}
\ea
\ba
\label{phi1_N}\dot{\widehat\phi}^{(1)}_t&=&2 q_2\widehat\phi^{(1)}_t+(\widehat\phi^{(1)}_t)^2 -(\epsilon_2-q_2^2),\\
  \nonumber  {\dot{\widehat\phi}^{(2)}_t}&=&2\left(-\widehat\eta^{(4)}_t+q_1\lambda_1(1-\beta_1)\right) \widehat\phi^{(2)}_t-(\widehat\phi_t^{(4)})^2\\
&&-2\left( \widehat\phi^{(4)}_t+q_2\lambda_2\beta_1\right)\widehat\phi^{(6)}_t-(\epsilon_2-q_2^2)\lambda_2^2\beta_2^2,\label{phi2_N}\\
\nonumber {\dot{\widehat\phi}^{(3)}_t} &=&2\left(-\widehat\phi^{(5)}_t+q_2\lambda_2(1-\beta_2)\right) \widehat\phi^{(3)}_t -(\phi_t^{(5)})^2\\
 &&-2\left(\widehat\eta^{(5)}_t+q_1\lambda_1\beta_2\right)\widehat \phi^{(6)}_t-(\epsilon_2-q_2^2)\lambda_2^2(\beta_2-1)^2,
\label{phi3_N}\\
 \dot{\widehat\phi}^{(4)}_t&=&\left(q_2-\widehat\eta^{(4)}_t-\widehat\phi_t^{(5)}+q_1\lambda_1(1-\beta_1)\right)\widehat\phi^{(4)}_t-q_2\lambda_2\beta_1\widehat \phi^{(5)}_t-(\epsilon_2-q_2^2)\lambda_2\beta_1,\label{phi4_N}\\
\nonumber \dot{\widehat\phi}^{(5)}_t&=&q_2\left(1+\lambda_2(1-\beta_2)\right)\widehat\phi^{(5)}_t-(\widehat\phi^{(5)}_t)^2\\
 &&-\left(\widehat\eta^{(5)}_t+q_1\lambda_1\beta_2\right)\widehat\phi_t^{(4)}+(\epsilon_2-q_2^2)\lambda_2(1-\beta_2),\label{phi5_N}\\
\nonumber\dot{\widehat\phi}^{(6)}_t&=&\left(-\widehat\eta^{(4)}_t-\widehat\phi^{(5)}_t+q_1\lambda_1(1-\beta_1)+q_2\lambda_2(1-\beta_2)\right) \widehat\phi^{(6)}_t-\widehat\phi_t^{(4)}\widehat\phi_t^{(5)}\\
\nonumber &&-\left(\widehat\eta^{(5)}_t+q_1\lambda_1\beta_2\right)\widehat \phi^{(2)}_t-\left(\widehat\phi^{(4)}_t+q_2\lambda_2\beta_1\right)\widehat\phi^{(3)}_t- (\epsilon_2-q_2^2)\lambda_2^2\beta_1(\beta_2-1),\\\label{phi6_N}
 \ea
with terminal conditions  
\ban
&&\widehat\eta_T^{(1)}=c_1,\quad  \widehat\eta_T^{(2)}=c_1\lambda_1^2(\beta_1-1)^2, \quad \widehat\eta_T^{(3)}=c_1\lambda_1^2\beta_2^2, \quad\widehat \eta_T^{(4)}=c_1\lambda_1(\beta_1-1),\\
&&\widehat\eta_T^{(5)}=c_1\lambda_1\beta_2,\quad\widehat \eta_T^{(6)}=c_1\lambda_1^2(\beta_1-1)\beta_2.
\ean
and 
\ban
&&\widehat\phi_T^{(1)}=c_2,\quad \widehat \phi_T^{(2)}=c_2\lambda_2^2\beta_1^2, \quad \widehat\phi_T^{(3)}=c_2\lambda_2^2(\beta_2-1)^2, \quad \widehat\phi_T^{(4)}=c_2\lambda_2\beta_1,\\
&&\widehat\phi_T^{(5)}=c_2\lambda_2(\beta_2-1),\quad \widehat\phi_T^{(6)}=c_2\lambda_2^2\beta_1(\beta_2-1).
\ean
where $0<\beta_1,\;\beta_2<1$, $\beta_1+\beta_2=1$, $0<\lambda_1,\;\lambda_2<1$, and $q_1,q_2,\epsilon_1,\epsilon_2>0$ with $\epsilon_1-q_1^2>0$ and $\epsilon_2-q_2^2>0$. The existence of the coupled equations \eqref{eta1_N} to \eqref{phi6_N} is verified. 
\end{prop}
\begin{proof}
We first observe that the existence of the coupled equations (\ref{eta1_N}-\ref{phi6_N}) is rely on the existence of the coupled Riccati equations (\ref{eta4_N}-\ref{eta5_N}) and (\ref{phi4_N}-\ref{phi5_N}). Using (\ref{eta4_N}-\ref{eta5_N}) and (\ref{phi4_N}-\ref{phi5_N}), we have 
\ba
\nonumber \dot{\widehat\eta_t^{(4)}}+\dot{\widehat\eta_t^{(5)}}&=&q_1\widehat\eta^{(4)}_t-(\widehat\eta^{(4)}_t)^2-\widehat\phi^{(4)}_t\widehat\eta^{(5)}_t+q_1\widehat\eta^{(5)}_t-\widehat\phi^{(5)}_t\widehat\eta^{(5)}_t-\widehat\eta^{(4)}_t\widehat\eta^{(5)}_t\\
&=&\left(q_1-\widehat\eta^{(4)}_t\right)\left(\widehat\eta^{(4)}_t+\widehat\eta^{(5)}_t\right)-\widehat\eta^{(5)}_t\left(\widehat\phi^{(4)}_t+\widehat\phi^{(5)}_t\right)\label{eqn_sum_1}
\ea
and similarly 
\ba
\dot{\widehat\phi_t^{(4)}}+\dot{\widehat\phi_t^{(5)}}&=&\left(q_2-\widehat\phi^{(5)}_t\right)\left(\widehat\phi^{(4)}_t+\widehat\phi^{(5)}_t\right)-\widehat\phi^{(4)}_t\left(\widehat\eta^{(4)}_t+\widehat\eta^{(5)}_t\right)\label{eqn_sum_2}
\ea
with the terminal conditions $\widehat\eta^{(4)}_T+\widehat\eta^{(5)}_T=0$ and $\widehat\phi^{(4)}_T+\widehat\phi^{(5)}_T=0$.  Observe that \eqref{eqn_sum_1} and \eqref{eqn_sum_2} are linear equations for $\widehat\eta^{(4)}_t+\widehat\eta^{(5)}_t$ and $\widehat\phi^{(4)}_t+\widehat\phi^{(5)}_t$ with the terminal condition being zero implying 
\be\label{suff_cond_1}
\widehat\eta^{(4)}_t+\widehat\eta^{(5)}_t=0, \quad \widehat\phi^{(4)}_t+\widehat\phi^{(5)}_t=0,
\en
for $0 \leq t\leq T$. Namely $\widehat\eta^{(4)}_t=-\widehat\eta^{(5)}_t$ and $\widehat\phi^{(4)}_t=-\widehat\phi^{(5)}_t$ for $0\leq t\leq T$.     
Hence, it is sufficient to study the existence of $\widehat\eta^{(5)}_t$ and $\widehat\phi^{(4)}_t$. 

Inserting $\widehat\eta^{(4)}_t=-\widehat\eta^{(5)}_t$ and $\widehat\phi^{(5)}_t=-\widehat\phi^{(4)}_t$ into $\widehat\eta^{(5)}_t$ and $\widehat\phi^{(4)}_t$ gives 
\ban
\nonumber\dot{\widehat\eta}_t^{(5)}&=&\left(q_1+\widehat\eta_t^{(5)}+\widehat\phi_t^{(4)}+q_2\lambda_2\beta_1\right)\widehat\eta_t^{(5)}+q_1\lambda_1\beta_2\widehat\eta_t^{(5)}-(\epsilon_1-q_1^2)\lambda_1\beta_2\\
&=&\left(q_1+q_2\lambda_2\beta_1+q_1\lambda_1\beta_2\right)\widehat\eta_t^{(5)}+(\widehat\eta_t^{(5)})^2+\widehat\phi_t^{(4)}\widehat\eta_t^{(5)}-(\epsilon_1-q_1^2)\lambda_1\beta_2,
\ean
and  
\ban
\dot{\widehat\phi}_t^{(4)}&=&\left(q_2+q_1\lambda_1\beta_2+q_2\lambda_2\beta_1\right)\widehat\phi_t^{(4)}+(\widehat\phi_t^{(4)})^2+\widehat\eta_t^{(5)}\widehat\phi_t^{(4)}-(\epsilon_2-q_2^2)\lambda_2\beta_1.
\ean
Now, consider ${\check\eta}_t^{(5)}=\widehat\eta_{T-t}^{(5)}$ and $\check\phi_t^{(4)}=\widehat\phi_{T-t}^{(4)}$. Then
\ba
\dot{\check\eta}_t^{(5)}&=&-\left(q_1+q_2\lambda_2\beta_1+q_1\lambda_1\beta_2\right)\check\eta_t^{(5)}-(\check\eta_t^{(5)})^2-\check\phi_t^{(4)}\check\eta_t^{(5)}+(\epsilon_1-q_1^2)\lambda_1\beta_2,\\
\dot{\check\phi}_t^{(4)}&=&-\left(q_2+q_1\lambda_1\beta_2+q_2\lambda_2\beta_1\right)\check\phi_t^{(4)}-(\check\phi_t^{(4)})^2-\check\eta_t^{(5)}\check\phi_t^{(4)}+(\epsilon_2-q_2^2)\lambda_2\beta_1,
\ea
with the initial conditions $\check\eta_0^{(5)}=c_1\lambda_1\beta_2$ and $\check\phi_0^{(4)}=c_2\lambda_2\beta_1$. Simple argument implies $\check\eta_t^{(5)}$ and $\check\phi_t^{(4)}$ being positive for $0\leq t\leq T$. Then, we have 
\ban
\dot{\check\eta}_t^{(5)}&\leq& -\left(q_1+q_2\lambda_2\beta_1+q_1\lambda_1\beta_2\right)\check\eta_t^{(5)}+(\epsilon_1-q_1^2)\lambda_1\beta_2,\\
\dot{\check\phi}_t^{(4)}&\leq& -\left(q_2+q_1\lambda_1\beta_2+q_2\lambda_2\beta_1\right)\check\phi_t^{(4)}+(\epsilon_2-q_2^2)\lambda_2\beta_1,
\ean
 leading to 
 \ban
 e^{\left(q_1+q_2\lambda_2\beta_1+q_1\lambda_1\beta_2\right)t}\check\eta_t^{(5)}&\leq&c_1\lambda_1\beta_2+(\epsilon_1-q_1^2)\lambda_1\beta_2\int_0^t e^{\left(q_1+q_2\lambda_2\beta_1+q_1\lambda_1\beta_2\right)s}ds\\
 &\leq&c_1\lambda_1\beta_2+\frac{(\epsilon_1-q_1^2)\lambda_1\beta_2}{q_1+q_2\lambda_2\beta_1+q_1\lambda_1\beta_2}e^{\left(q_1+q_2\lambda_2\beta_1+q_1\lambda_1\beta_2\right)t}\\
 \ean
such that 
\be
0\leq \check\eta_t^{(5)} \leq c_1\lambda_1\beta_2 e^{-\left(q_1+q_2\lambda_2\beta_1+q_1\lambda_1\beta_2\right)t}+\frac{(\epsilon_1-q_1^2)\lambda_1\beta_2}{q_1+q_2\lambda_2\beta_1+q_1\lambda_1\beta_2}.
\en
Similarly, we also get
\be
0\leq \check\phi_t^{(4)}\leq c_2\lambda_2\beta_1 e^{-\left(q_2+q_1\lambda_1\beta_2+q_2\lambda_2\beta_1\right)t}+\frac{(\epsilon_2-q_2^2)\lambda_2\beta_1}{q_2+q_1\lambda_1\beta_2+q_2\lambda_2\beta_1}. 
\en
The proof is complete. 
\end{proof}

According to the results in Proposition \ref{Prop_suff}, as $N_1$ and $N_2$ are large enough,  the existence of the coupled ODEs  \eqref{eta1} to \eqref{phi10} are guaranteed.

\subsection{Open-loop Nash Equilibria}\label{FBSDE-approach}
Referring to \cite{R.Carmona2013}, we now study the open-loop Nash equilibria where the strategies are the functions given at initial time. Namely, the strategies are the function of $t$ and $X_0$ given in \cite{CarmonaSIAM2016}. 
\begin{theorem}\label{Hete-open}
Assume $q_k^2\leq \epsilon_k$ for $k=1,2$ and $\eta^{o,(i)}_t$ and $\phi^{o,(i)}_t$ for $i=1,\cdots,4$ satisfying \eqref{eta_open-1} to \eqref{phi_open-4}. The open-loop Nash equilibria are written as 
\ba
\nonumber \hat\alpha^{o,(1)i}&=&\left(q_1+(1-\frac{1}{\widetilde N_1})\eta_t^{o,(1)}\right)(\overline X_t^{(1)}-X_t^{(1)i})\\
\nonumber&&+\left(q_1\lambda_1(\beta_1-1)+(1-\frac{1}{\widetilde N_1})\eta_t^{o,(2)}\right)\overline X_t^{(1)}\\
&&+\left(q_1\lambda_1\beta_2+(1-\frac{1}{\widetilde N_1})\eta_t^{o,(3)}\right)\overline X_t^{(2)}+\left(1-\frac{1}{\widetilde N_1}\right)\eta_t^{o,(4)},\\
\nonumber \hat\alpha^{o,(2)j}&=&\left(q_2+(1-\frac{1}{\widetilde N_2})\phi_t^{o,(1)}\right)(\overline X_t^{(2)}-X_t^{(2)j})\\
\nonumber &&+\left(q_2\lambda_2\beta_1+(1-\frac{1}{\widetilde N_2})\phi_t^{o,(2)}\right)\overline X_t^{(1)}\\
&&+\left(q_2\lambda_2(\beta_2-1)+(1-\frac{1}{\widetilde N_2})\phi_t^{o,(3)}\right)\overline X_t^{(2)}+\left(1-\frac{1}{\widetilde N_2}\right)\phi_t^{o,(4)},
\ea
 where $$\frac{1}{\widetilde N_k}=\frac{1-\lambda_k}{N_k}+\frac{\lambda_k}{N},$$ for $k=1,2$. 
\end{theorem}
\begin{proof}
See Appendix \ref{Appex-open}. 
\end{proof}
Due to the complexity of coupled ordinary differential equations (ODEs) (\ref{eta_open-1}-\ref{phi_open-4}), we also verify the existence of (\ref{eta_open-1}-\ref{phi_open-4}) in the case of $N_1\rightarrow\infty$ and $N_2\rightarrow\infty$ where the equations are given by
\ba
\dot{\widehat{\eta}}^{o,(1)}_t&=&2 q_1\widehat\eta^{o,(1)}_t+ (\widehat\eta_t^{o,(1)})^2-(\epsilon_1-q_1^2),\\
\nonumber \dot{\widehat\eta}^{o,(2)}_t&=&q_1\left(1+\lambda_1(1-\beta_1)\right) \widehat\eta^{o,(2)}_t-(\widehat\eta^{o,(2)}_t)^2\\
&&-\left(\widehat\phi^{o,(2)}_t+q_2\lambda_2\beta_1\right)\widehat\eta^{o,(3)}_t+(\epsilon_1-q_1^2)\lambda_1(1- \beta_1),\\
\nonumber  \dot{\widehat\eta}^{o,(3)}_t&=&\left(q_1-\widehat\eta^{o,(2)}_t-\widehat\phi^{o,(3)}_t+q_2\lambda_2(1-\beta_2)\right) \widehat\eta^{o,(3)}_t\\
  &&-q_1\lambda_1\beta_2\widehat\eta^{o,(2)}_t-(\epsilon_1-q_1^2)\lambda_1\beta_2,\\
\dot{\widehat\phi}^{o,(1)}_t&=&2 q_2\widehat\phi^{o,(1)}_t+(\widehat\phi^{o,(1)}_t)^2 -(\epsilon_2-q_2^2),\\
  \nonumber\dot{\widehat\phi}^{o,(2)}_t&=&\left(q_2-\widehat\eta^{o,(2)}_t-\widehat\phi_t^{o,(3)}+q_1\lambda_1(1-\beta_1)\right)\widehat\phi^{o,(2)}_t\\
&&  -q_2\lambda_2\beta_1\widehat \phi^{o,(3)}_t-(\epsilon_2-q_2^2)\lambda_2\beta_1, \\
\nonumber \dot{\widehat\phi}^{o,(3)}_t&=&q_2\left(1+\lambda_2(1-\beta_2)\right)\widehat\phi^{o,(3)}_t-(\widehat\phi^{o,(3)}_t)^2\\
 &&-\left(\widehat\eta^{o,(3)}_t+q_1\lambda_1\beta_2\right)\widehat\phi_t^{o,(2)}+(\epsilon_2-q_2^2)\lambda_2(1-\beta_2),
\ea
with the terminal conditions 
\ban
\widehat\eta_T^{o,(1)}=c_1,\;\widehat\eta_T^{o,(2)}=c_1\lambda_1(\beta_1-1),\;\widehat\eta_T^{o,(3)}=c_1\lambda_1\beta_2,\\
\widehat\phi_T^{o,(1)}=c_2,\;\widehat\eta_T^{o,(2)}=c_2\lambda_2\beta_1,\;\widehat\phi_T^{o,(3)}=c_2\lambda_2(\beta_1-1). 
\ean
Note that according to the results in Section \ref{HJB-approach}, we obtain the Riccati equations $\widehat\eta_T^{o,(2)}=\widehat\eta_T^{(4)}$ and $\widehat\phi_T^{o,(3)}=\widehat\phi_T^{(5)}$ and the linear ODEs $\widehat\eta_T^{o,(3)}=\widehat\eta_T^{(5)}$ and $\widehat\phi_T^{o,(2)}=\widehat\phi_T^{(4)}$. Hence, the existence of the coupled ODEs is verified by Proposition \ref{Prop_suff}.

 We then discuss  $\epsilon$-Nash Equilibria in the case of the general $d$ groups mean field game. Namely, $N\rightarrow\infty$, $N_k\rightarrow\infty$ with $\frac{N_k}{N}\rightarrow\beta_k$ for $k=1,\cdots,d$.

\section{$\epsilon$-Nash Equilibria: Mean Field Games}\label{MFG}
We return to the case of the general $d$ groups. The $i$-th bank in group $k$ minimizes the objective function given by
\ba 
\nonumber J^{(k)i}(\alpha)&=&\EE\bigg\{\int_{0}^{T}
 \left[\frac{(\alpha_{t}^{(k)i})^2}{2}-q_k\alpha_{t}^{(k)i}\left(\overline{X}^{\lambda_k}_{t}-X^{(k)i}_{t}\right)
+\frac{\epsilon_k}{2}\left(\overline{X}^{\lambda_k}_{t}-X^{(k)i}_{t}\right)^2\right]dt\\
& &+ \frac{c_k}{2}\left(\overline{X}^{\lambda_k}_{T}-X^{(k)i}_{T}\right)^2\bigg\},
\ea
with $q_k^2<\epsilon_k$ subject to 
\ba
\nonumber dX^{(k)i}_t &=& (\alpha_{t}^{(k)i}+\gamma^{(k)}_{t})dt\\
 &&+\sigma_k\left(\rho dW^{(0)}_t+\sqrt{1-\rho^2}\left(\rho_{k}dW_t^{(k)}+\sqrt{1-\rho^2_{k}}dW^{(k)i}_t\right)\right), 
\ea
with the initial value $X_0^{(k)i}$ which may also be a squared integrable random variable $\xi^{(k)}$. In the mean field limit as $N\rightarrow \infty$ and $\frac{N_k}{N}\rightarrow\beta_k$ for all $k$, as in the case of one group, the $d$ heterogenous groups is solved by the $d$-players game. Referred to \cite{R.Carmona2013}, the scheme to solve the $\epsilon$-Nash equilibria is as follows:
\begin{enumerate}
\item Fix $m^{(k)}_t=\EE[X^{(k)}_t|(W^{(0)}_s)_{s\leq t}]$  which is a candidate for the limit of $\overline{X}^{(k)}_t$ as $N_k\rightarrow \infty$:
    \[
    m^{(k)}_t=\lim_{I_k\rightarrow\infty}\overline{X}^{(k)}_t,
    \]
for all $k$, and
    \[
    M_t=\lim_{N,N_1\ldots,N_k\rightarrow\infty}\sum_{k=1}^d\frac{N_k}{N}\overline{X}^{(k)}_t=\sum_{k=1}^d\beta_km^{(k)}_t,
    \]
  where a vector of standard Brownian motions $W^{(0)}=(W^{(0),(0)},\cdots,W^{(0),(d)})$ represents the common noises. 
\item Substitute $m_t^{(k)}$ to $\overline X^{(k)}$ and solve the $d$-players control problem through minimizing the objective function written as
\ban
&& \nonumber  \inf_{ \alpha^{(k)}\in {\cal{A}}} \EE\bigg\{\int_{0}^{T} \left[\frac{(\alpha_{t}^{(k)})^2}{2}-q_k\alpha_{t}^{(k)}\left(M^{\lambda_k}_{t}-X^{(k)}_{t}\right)+\frac{\epsilon_k}{2}\left(M^{\lambda_k}_{t}-X^{(k)}_{t}\right)^2\right]dt\\
&&\quad\quad\quad+ \frac{c_k}{2}\left(M^{\lambda_k}_{T}-X^{(k)}_{T}\right)^2\bigg\}
\ean
with $q_k^2<\epsilon_k$  subject to the dynamics
      \ba
\nonumber dX^{(k)}_t &=& (\alpha_{t}^{(k)}+\gamma^{(k)}_{t})dt\\
\nonumber &&+\sigma_k\left(\rho dW^{(0),(0)}_t+\sqrt{1-\rho^2}\left(\rho_{k}dW_t^{(0),(k)}+\sqrt{1-\rho^2_{k}}dW^{(k)}_t\right)\right),\\\label{X-d-groups}
\ea
with $M_t^{\lambda_k}=(1-\lambda_k)m_t^{(k)}+\lambda_kM_t$ where $W_t^{(k)}$ is a Brownian motion independent of $W_t^{(0)}$. 
\item Similarly, solve the fixed point problem: find $$m^{(k)}_t=\EE[X^{(k)}_t|(W^{(0)}_s)_{s\leq t},(W_s^{(0),(k)})_{s\leq t}]$$ for all $t$.
\end{enumerate}

\begin{theorem}\label{Hete-MFG-prop}
Assuming $q_k^2\leq \epsilon_k$ and $\eta_t^{m,(k)}$, $\psi^{m,(k),h}_t$, and $\mu^{m,(k)}_t$ satisfying 
\ba
\label{Hete-eta-MFG}\dot\eta_t^{m,(k)}&=&2q_k \eta_t^{m,(k)}+(\eta_t^{m,(k)})^2-(\epsilon_k-q_k^2),\\
\nonumber  \dot\psi_t^{m,(k),h_1}&=&q_k\psi_t^{m,(k),h_1}-\sum_{h=1}^d\psi_t^{m,(k),h}\left(\psi_t^{m,(h),h_1}+q_h\lambda_h(\beta_{h_1 }-\delta_{h,h_1})\right)\\
&&-(\epsilon_k-q_k^2)\lambda_k(\beta_{h_1}-\delta_{k,h_1})\label{Hete-psi-MFG}\\
 \dot\mu^{m,(k)}_t&=&q_k\mu^{m,(k)}_t-\sum_{h=1}^d\psi_t^{m,(k),h}(\mu_t^{m,(h)}+\gamma_t^{(h)}),
\label{Hete-mu-MFG}
\ea
with terminal conditions $\eta_T^{m,(k)}=c_k$, $\psi_T^{m,(k),h}=c_k\lambda_k(\beta_h-\delta_{k,h})$, and $\mu^{m,(k)}_T=0$ for $k,h=1,\cdots,d$, the $\epsilon$-Nash equilibrium is given by
\be\label{Hete-optimal-MFG}
\hat\alpha_t^{m,(k)}=(q_k+\eta^{m,(k)}_t)( m^{(k)}_t-x^{(k)})+\sum_{h=1}^d\widetilde\psi_t^{m,(k),h}m^{(h)}_t+\mu^{m,(k)}_t,\quad 
\en
where $m_t^{(k)}$ is given by
\ban
\nonumber dm^{(k)}_t&=&\bigg\{\sum_{h_1=1}^d\psi_t^{m,(k),h_1}m^{(h_1)}_t+\mu^{m,(k)}_t+\gamma^{(k)}_t+q_k\lambda_k\sum_{h_1=1}^d(\beta_{h_1}-\delta_{k,h_1})m^{(h_1)}_t\bigg\}dt\\
&&+\sigma_k\left(\rho dW^{(0),(0)}_t+\sqrt{1-\rho^2} \rho_{k}dW_t^{(0),(k)} \right)
\ean
for $k=1,\cdots,d$. Given $c_{\tilde k}\geq \max_{k,h}\left(\frac{q_k\lambda_k}{\lambda_h}-q_h\right)$ for $\tilde k=1,\cdots,d$, the existence of the coupled ODEs \eqref{Hete-eta-MFG} to \eqref{Hete-mu-MFG} can be verified.
\end{theorem}
\begin{proof}
See Appendix \ref{Appex-Hete-MFG}. 
\end{proof}
In particular, in the case of $d=2$, as $N\rightarrow\infty$, $N_1\rightarrow\infty$, and $N_2\rightarrow\infty$ with $\frac{N_1}{N}\rightarrow\beta_1$ and $\frac{N_2}{N}\rightarrow\beta_2$, we observe that $\widetilde\eta^{(4)}_t\rightarrow\widetilde\psi_t^{m,(1),1}$, $\widetilde\eta^{(5)}_t\rightarrow\widetilde\psi_t^{m,(1),2}$, $\widetilde\eta_t^{(7)}\rightarrow\widetilde\mu_t^{m,(1)}$, $\widetilde\phi^{(4)}_t\rightarrow\widetilde\psi_t^{m,(2),1}$, $\widetilde\phi^{(5)}_t\rightarrow\psi_t^{m,(2),2}$, and  $\widetilde\phi_t^{(7)}\rightarrow\widetilde\mu_t^{m,(2)}$ for $0\leq t \leq T$ such that the $\epsilon$-Nash equilibria in the case of the mean field game with heterogeneous groups are the limit of the closed-loop and the open loop Nash equilibria in the case of the finite player game with heterogeneous groups. The results are consistent with \cite{Lacker2018}. Hence, the asymptotic optimal strategy for bank $(1)i$ is given by
\[
\hat\alpha^{m,(1)i}_t=(q_1+ \eta^{m,(1)}_t)( \overline x^{(1)}-x^{(1)i})+\widetilde\psi_t^{m,(1),1}\overline x^{(1)}+\widetilde\psi_t^{m,(1),2}\overline x^{(2)}+\mu^{m,(1)}_t. 
\]

\begin{remark}
The case of no common noise implies the given $m_t$ for $0\leq t \leq T$ being a deterministic function. For instance, in the case of $d=1$, the $\epsilon$-Nash equilibrium in mean field games can be obtained using the HJB equation written as
\ban
\partial_tV&+&\inf_{\alpha}\bigg\{\alpha\partial_xV+\frac{\sigma^2}{2}\partial_{xx}V +\frac{\alpha^2}{2}-q\alpha(m_t-x)+\frac{\epsilon}{2}(m_t-x)^2\bigg\}=0, \\
\ean
 with the terminal condition $V(T,x)=\frac{c}{2}(m_T-x)^2$. 
\end{remark}


\section{Financial Implications}\label{FI}
The aim of this section is to investigate the financial implications for this heterogeneous interbank lending and borrowing model.  We mainly comment on the closed-loop Nash equilibria in the case of finite players identical to the open-loop Nash equilibria and $\epsilon$-Nash equilibria. We recall the closed-loop Nash equilibria written as 
\ba
\label{optimal-finite-ansatz-V1-FI}
\hat\alpha^{(1)i}(t,x)&=&(q_1+\widetilde\eta^{(1)}_t)( \overline x^{(1)}-x^{(1)i})+\widetilde\eta^{(4)}_t\overline x^{(1)}+\widetilde\eta^{(5)}_t\overline x^{(2)}+\widetilde\eta_t^{(7)},\\
\hat\alpha^{(2)j}(t,x)&=&(q_2+\widetilde\phi^{(1)}_t)(\overline x^{(2)}-x^{(2)j})+\widetilde\phi^{(4)}_t\overline x^{(1)}+\widetilde\phi^{(5)}_t\overline x^{(2)}+\widetilde\phi_t^{(7)}.\label{optimal-finite-ansatz-V2-FI}
\ea
The corresponding optimal trajectories are written as
\ba
\nonumber dX_t^{(1)i}&=&\left\{(q_1+\widetilde\eta^{(1)}_t)( \overline X_t^{(1)}-X_t^{(1)i})+\widetilde\eta^{(4)}_t\overline X_t^{(1)}+\widetilde\eta^{(5)}_t\overline X_t^{(2)}+\widetilde\eta_t^{(7)}+\gamma^{(1)}_t\right\}dt \\
\nonumber&& +\sigma_1\left(\rho dW^{(0)}_t+\sqrt{1-\rho^2}\left(\rho_{1}dW_t^{(1)}+\sqrt{1-\rho^2_{1}}dW^{(1)i}_t\right)\right)\\
\nonumber &=&\left\{\frac{q_1+\widetilde\eta^{(1)}_t}{N_1}\sum_{l=1}^{N_1}(X_t^{(1)l}-X_t^{(1)i})+\widetilde\eta^{(4)}_t\overline X_t^{(1)}+\widetilde\eta^{(5)}_t\overline X_t^{(2)}+\widetilde\eta_t^{(7)}+\gamma^{(1)}_t\right\}dt \\
\label{X-1-FI}&& +\sigma_1\left(\rho dW^{(0)}_t+\sqrt{1-\rho^2}\left(\rho_{1}dW_t^{(1)}+\sqrt{1-\rho^2_{1}}dW^{(1)i}_t\right)\right),\\
\nonumber  dX_t^{(2)j}&=&\left\{(q_2+\widetilde\phi^{(1)}_t)( \overline X_t^{(2)}-X_t^{(2)j})+\widetilde\phi^{(4)}_t\overline X_t^{(1)}+\widetilde\phi^{(5)}_t\overline X_t^{(2)}+\widetilde\phi_t^{(7)}+\gamma^{(2)}_t\right\}dt\\
\nonumber&& +\sigma_2\left(\rho dW^{(0)}_t+\sqrt{1-\rho^2}\left(\rho_{2}dW_t^{(2)}+\sqrt{1-\rho^2_{2}}dW^{(2)i}_t\right)\right)\\
\nonumber&=&\left\{\frac{q_2+\widetilde\phi^{(1)}_t}{N_2}\sum_{l=1}^{N_2}(X^{(2)l}-X_t^{(2)j})+\widetilde\phi^{(4)}_t\overline X_t^{(1)}+\widetilde\phi^{(5)}_t\overline X_t^{(2)}+\widetilde\phi_t^{(7)}+\gamma^{(2)}_t\right\}dt\\
\label{X-2-FI}&& +\sigma_2\left(\rho dW^{(0)}_t+\sqrt{1-\rho^2}\left(\rho_{2}dW_t^{(2)}+\sqrt{1-\rho^2_{2}}dW^{(2)i}_t\right)\right)
\ea
with given initial values $X_0^{(1)i}$ and $X_0^{(2)j}$ for $i=1,\cdots,N_1$ and $j=1,\cdots,N_2$.  
Compared to the homogeneous group scenario discussed in \cite{R.Carmona2013}, owing to the linear quadratic structure, heterogeneity implies that banks not only purely consider the distance between their capitalization and averages of their own group capitalization where these terms can be rewritten as
$
\frac{1}{N_1}\sum_{l=1}^{N_1}(X_t^{(1)l}-X_t^{(1)i})
$
and $\frac{1}{N_2}\sum_{l=1}^{N_2}(X^{(2)l}-X_t^{(2)j})$ show the lending and borrowing behavior within their own groups
identical to the homogeneous case but also the linear combination of the average of each group.

In particular, as $N_1$ and $N_2$ are sufficiently large , based on the relation $\widetilde\eta^{(4)}_t=-\widetilde\eta^{(5)}_t$ and $\widetilde\phi^{(4)}_t=-\widetilde\phi^{(5)}_t$ in Proposition \ref{Prop_suff}, we rewrite the system (\ref{X-1-FI}-\ref{X-2-FI}) as 
\ba
\nonumber dX_t^{(1)i}&=&\left\{\frac{q_1+\widetilde\eta^{(1)}_t}{N_1}\sum_{l=1}^{N_1}(X_t^{(1)l}-X_t^{(1)i})+\widetilde\eta^{(5)}_t(\overline X_t^{(2)}-\overline X_t^{(1)})+\widetilde\eta_t^{(7)}+\gamma^{(1)}_t\right\}dt \\
\label{X-1-FI-1}&& +\sigma_1\left(\rho dW^{(0)}_t+\sqrt{1-\rho^2}\left(\rho_{1}dW_t^{(1)}+\sqrt{1-\rho^2_{1}}dW^{(1)i}_t\right)\right),\\
\nonumber  dX_t^{(2)j}&=&\left\{\frac{q_2+\widetilde\phi^{(1)}_t}{N_2}\sum_{l=1}^{N_2}(X^{(2)l}-X_t^{(2)j})+\widetilde\phi^{(4)}_t(\overline X_t^{(1)}-\overline X_t^{(2)})+\widetilde\phi_t^{(7)}+\gamma^{(2)}_t\right\}dt\\
\label{X-2-FI-1}&& +\sigma_2\left(\rho dW^{(0)}_t+\sqrt{1-\rho^2}\left(\rho_{2}dW_t^{(2)}+\sqrt{1-\rho^2_{2}}dW^{(2)i}_t\right)\right).
\ea
The terms $\widetilde\eta^{(5)}_t(\overline X_t^{(2)}-\overline X_t^{(1)})$ and $\widetilde\phi^{(4)}_t(\overline X_t^{(1)}-\overline X_t^{(2)})$ with $\widetilde\eta^{(5)}_t$ and $\widetilde\phi^{(4)}_t$ being positive for $0\leq t \leq T$ give the mean-reverting interaction between groups with each other such that the ensemble averages intend to be close. The dynamics of the distance $\overline X^{D}_t= \overline X^{(1)}_t-\overline X_t^{(2)} $ is written as
\ba
\nonumber d\overline X^{D}_t&=&-\left\{(\widetilde\eta^{(5)}_t+\widetilde\phi^{(4)}_t)\overline X^{D}_t+(\widetilde\eta_t^{(7)}+\gamma^{(1)}_t-\widetilde\phi_t^{(7)}-\gamma^{(2)}_t)\right\}dt\\
\nonumber &&+\rho\left(\sigma_1-\sigma_2\right)dW^{(0)}_t+\sqrt{1-\rho^2}\left(\sigma_1\rho_1dW^{(1)}_t-\sigma_2\rho_2dW_t^{(2)}\right)\\
&&+\sqrt{1-\rho^2}\left(\sqrt{1-\rho_1^2}\frac{1}{N_1}\sum_{l=1}^{N_1}dW^{(1)l}_t-\sqrt{1-\rho_2^2}\frac{1}{N_2}\sum_{l=1}^{N_2}dW^{(2)l}_t\right). 
\ea
with $\overline X_0^{(D)}=\overline X^{(1)}_0-\overline X_0^{(2)}$. As $N_1, N_2\rightarrow\infty$, the distance $\overline X^{D}_t$ is driven by common noises $W^{(0)}_t$, $W^{(1)}_t$ and $W^{(2)}_t$. Namely, the stronger correlation leads to larger fluctuation between groups.  

On the contrary, in the case of no common noise with $\rho=\rho_1=\rho_2=0$ and no growth rates $\gamma^{(1)}=\gamma^{(2)}_t=0$ leading to $\widetilde\eta_t^{(7)}=\widetilde\phi_t^{(7)}=0$ for all $t\geq 0$, we obtain $\overline X^{(D)}_t\rightarrow 0$ as $t\rightarrow\infty$ in the sense that in the long term behavior, all banks trace the global average $\overline X_t$ driven by only the scaled Brownian motion. This implies that systemic risk happens in the same manner as studied in \cite{R.Carmona2013} and therefore
\[
\PP(\tau<\infty) = \lim_{T\to\infty}\PP(\tau \leq T) =\lim_{T\to\infty}2\Phi\left(\frac{D\sqrt{N}}{\sigma\sqrt{T}}\right)= 1,
\]
with $\Phi$ being the standard normal distribution function  where $\tau=\inf\{t:\overline X_t\leq D\}$ with the certain default level $D$. The systemic event 
\[
\left\{\left(\frac{1}{N}\sum_{i=1}^{N}X_{t}^{(i)}\right)\leq {D}\; \mathrm{for\;some}\; t\right\} 
\]
defined by \cite{Fouque-Sun} is unavoidable.

In the numerical analysis, suppose that the first group is the group of stronger banks and the second one is the group of smaller banks. As discussion in Section \ref{Heter}, we first assume $\beta_1=0.2$ and $\beta_2=0.8$ and further fix the relative consideration $\lambda_1=0.1$ for the relative ensemble average
\[
\overline x^{\lambda_1}=(1-\lambda_1)\overline x^{(1)}+\lambda_1\overline{x},
\]
since the major players prefer tracing the group average rather than the ensemble average. By using varied $\lambda_2$ and $N$, we then obtain the following implications:

\begin{enumerate}
\item We first comment on two extreme cases. As $\lambda_1=\lambda_2=0$ or $\beta_1=\beta_2=1$, the model degenerates to the two homogeneous group model without the interaction between groups referred to \cite{R.Carmona2013}. 
\item In the intermediate region, in Figure \ref{liquidityratelambda}, we observe that the liquidity rate 
\be\label{liquidityrate}
\tilde\eta^{(1)}=(1-\frac{1}{N_1})\eta^{(1)}-\frac{1}{N_1}\eta^{(4)}
\en
increases in the relative proportion $\lambda_2$. Namely, when banks consider to trace the global ensemble average $\overline x$, they intend to lend to or borrow from a central bank more frequently.  
 \item As the terminal time $T$ becomes large, Figure \ref{liquidityratelambda_t=10} shows that the liquidity rate intends to be a constant. The identical results are obtained in \cite{R.Carmona2013}.  
\item As the numbers of banks  $N$, $N_1$, and $N_2$ become large, the liquid rate \eqref{liquidityrate} also increases. The interbank lending and borrowing behavior becomes more frequently.  See Figure \ref{liquidityrate_N} for instance.    
 
 \end{enumerate}

\begin{figure}[htbp]
\begin{center}
\includegraphics[width=12cm,height=8cm]{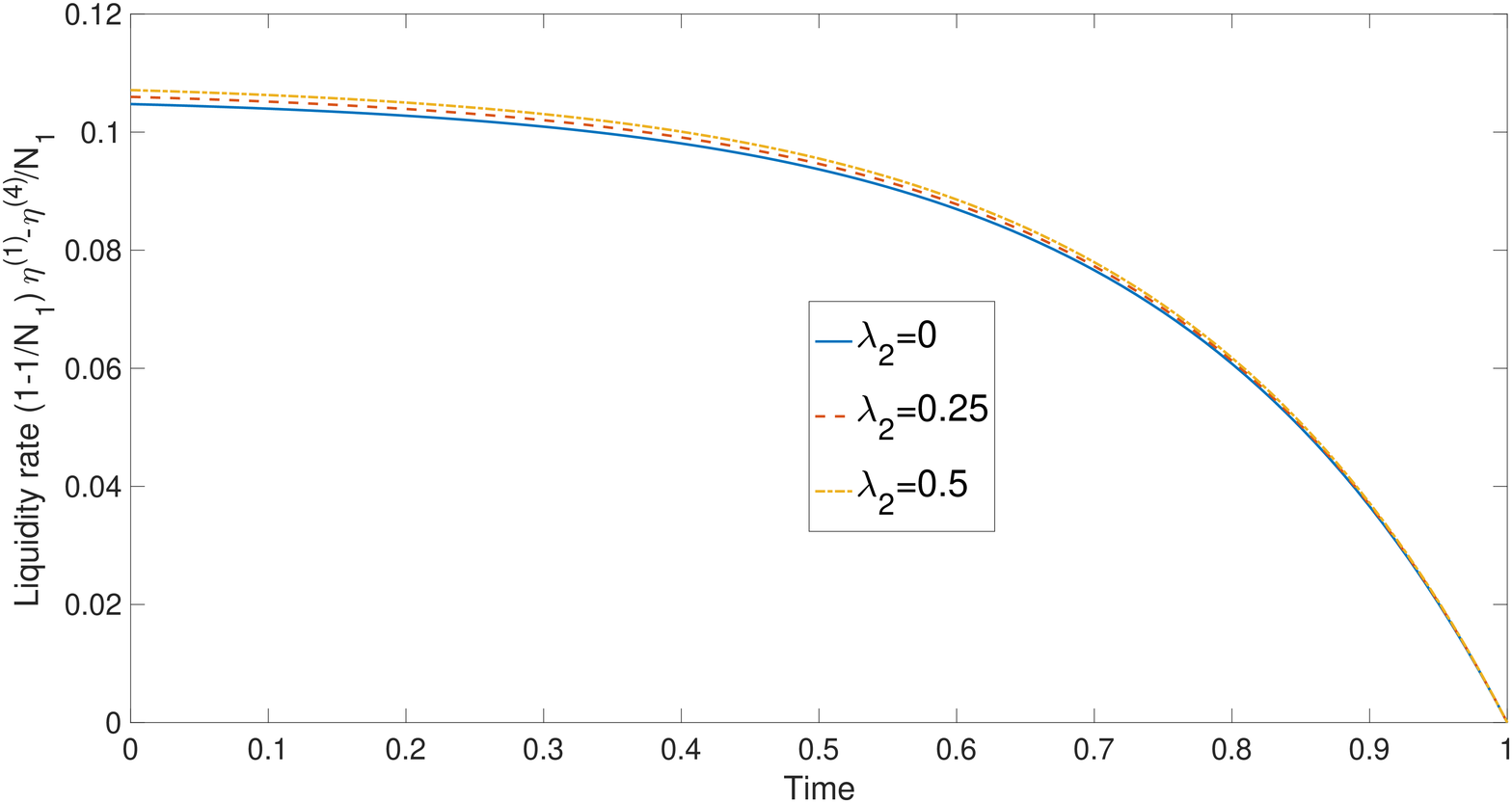}
\caption{The liquidity $\tilde\eta^{(1)}=(1-\frac{1}{N_1})\eta^{(1)}-\frac{1}{N_1}\eta^{(4)}$ with the varied $\lambda_2$ and the fixed $\lambda_1=0.1$.  The fixed parameters are $N=10$, $N_1=2$, $N_2=8$, $q_1=q_2=2$, $\epsilon_1=5$, $\epsilon_2=4.5$, $c_1=c_2=0$, $T=1$, and $\gamma^{(1)}_t=\gamma^{(2)}_t=0$ for $0\leq t\leq T$. }
\label{liquidityratelambda}
\end{center}
\end{figure}

\begin{figure}[htbp]
\begin{center}
\includegraphics[width=12cm,height=8cm]{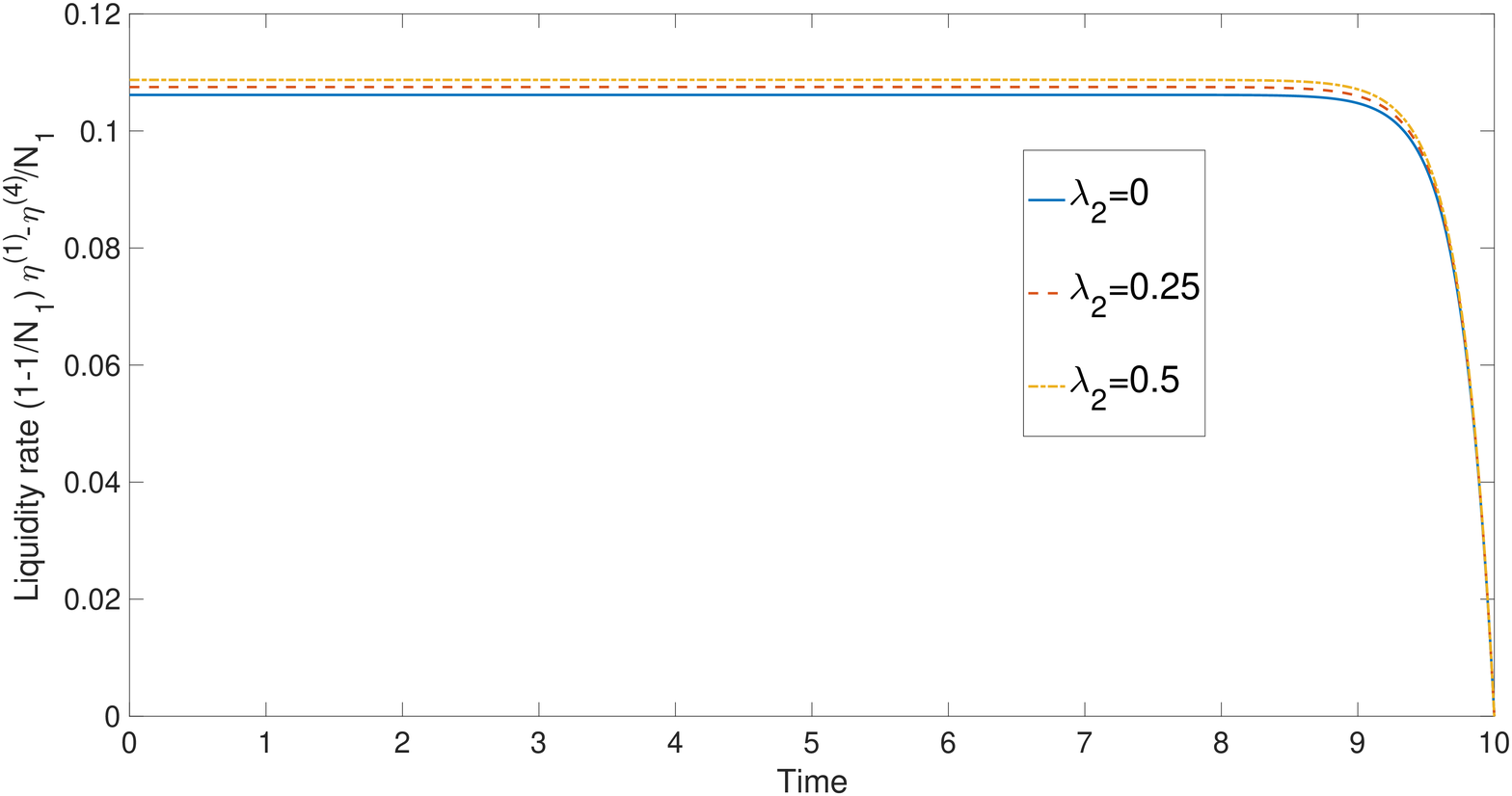}
\caption{The liquidity $\tilde\eta^{(1)}=(1-\frac{1}{N_1})\eta^{(1)}-\frac{1}{N_1}\eta^{(4)}$ with the varied $\lambda_2$ and the fixed $\lambda_1=0.1$.  The fixed parameters are $N=10$, $N_1=2$, $N_2=8$, $q_1=q_2=2$, $\epsilon_1=5$, $\epsilon_2=4.5$, $c_1=c_2=0$, $T=10$, and $\gamma^{(1)}_t=\gamma^{(2)}_t=0$ for $0\leq t\leq T$.  }
\label{liquidityratelambda_t=10}
\end{center}
\end{figure}

 \begin{figure}[htbp]
\begin{center}
\includegraphics[width=6cm,height=8cm]{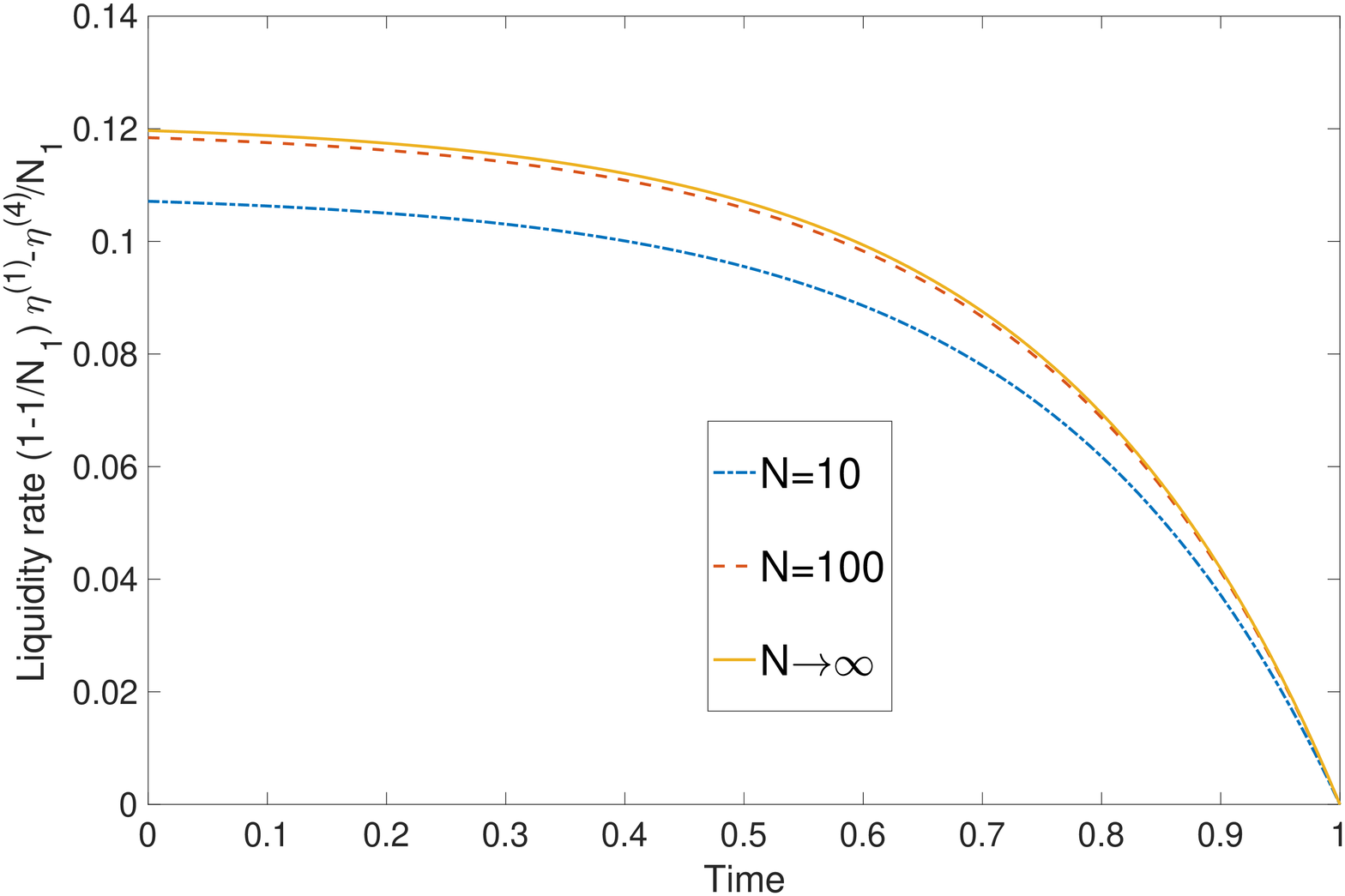}
\includegraphics[width=6cm,height=8cm]{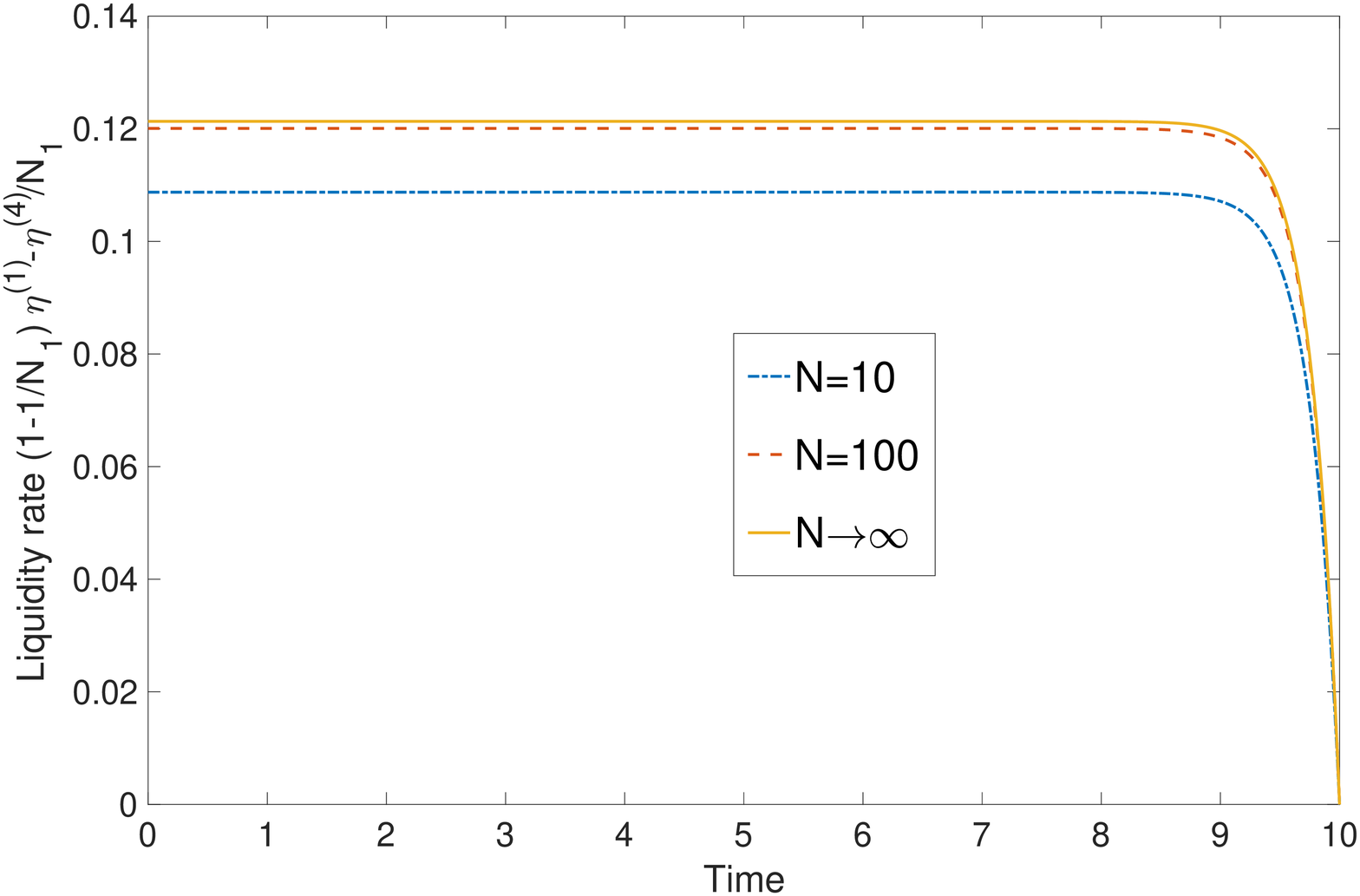}
\caption{The liquidity $\tilde\eta^{(1)}=(1-\frac{1}{N_1})\eta^{(1)}-\frac{1}{N_1}\eta^{(4)}$ with the varied $N$, $N_1$, $N_2$ with the proportion $\beta_1=0.2$, $\beta_2=0.8$. The terminal times are $T=1$ (left) and $T=10$ (right). The parameters are $\lambda_1=0.5$, $\lambda_2=0.1$, $q_1=q_2=2$, $\epsilon_1=5$, $\epsilon_2=4.5$, $c_1=c_2=0$, and $\gamma^{(1)}_t=\gamma^{(2)}_t=0$ for $0\leq t\leq T$.}
\label{liquidityrate_N}
\end{center}
\end{figure}

\section{Conclusions}\label{conclusions}
We study the system of interbank lending and borrowing with heterogeneous groups where the lending and borrowing depends on the homogeneous parameters within groups and heterogeneous parameters between groups. The amount of lending and borrowing is based on the relative ensemble average
\[
\tilde x^{\lambda_k}=(1-\lambda_k)\overline x^{(k)}+\lambda_k\overline{x},\quad k=1,\cdots,d. 
\]
Due to the heterogeneity structure, the value functions and the corresponding closed-loop and open-loop Nash equilibria are given by the coupled Riccati equations. In the two-group case, the existence of the coupled Riccati equations can be proved through when the number of banks are sufficiently large so that the existence of the value functions and equilibria are guaranteed. In addition, in the case of mean field games with the general $d$ groups, the existence of the $\epsilon$-Nash equilibria is also verified. 

We observe that owing to heterogeneity, the equilibria are consisted of the term of mean-reverting at their own group averages and all group ensemble averages . In the mean field case with no common noise, systemic event happens almost surely in the long time period. The numerical results illustrate that  as banks intend to trace the global average as large $\lambda_k$, they prefer liquidating more frequently using a larger liquidity rate. The liquidity rate is also increasing in the number of banks. 

The problem can be extended in several directions. First, it is interesting to discuss the delay obligation based on the model studied in \cite{Carmona-Fouque2016,FouqueZhang2018}. Second, it is nature to consider the stochastic growth rate in the system. Third, the CIR type processes can be applied to describe the capitalization of banks. See \cite{Fouque-Ichiba,Sun2016}. Furthermore, referring to \cite{BMMB2019}, the bubble assets is worth to study in the interbank lending and borrowing system. The admissible conditions for the equilibria of the above extensions are also interesting to investigate. 

\section*{Acknowledgements}
The author would like to thank  Jean-Pierre Fouque and Shuenn-Jyi Sheu for the conversations and suggestion on this subject.

 \appendix
 \numberwithin{equation}{section}

\makeatletter 
\newcommand{\section@cntformat}{Appendix \thesection:\ }
\makeatother

 \section{Proof of Theorem \ref{Hete-Nash} and Verification Theorem} \label{Appex-1}
The corresponding coupled HJB equations for the value functions (\ref{value-function-1}) and (\ref{value-function-2}) read 
\ba\label{HJB1}
\nonumber  \partial_{t}V^{(1)i} &+&\inf_{\alpha^{(1)i}}\bigg\{
\sum_{l\neq i,l=1}^{N_1}\bigg(\gamma^{(1)}_t+{\hat\alpha^{(1)l}(t,x)}\bigg)\partial_{x^{(1)l}}V^{(1)i}
+ \bigg(\gamma^{(1)}_t+{\alpha^i}\bigg)\partial_{x^{(1)i}}V^{(1)i}\\
\nonumber&+&\sum_{h=1}^{N_2}\bigg(\gamma_t^{(2)}+{\hat\alpha^{(2)h}(t,x)}\bigg)\partial_{x^{(2)h}}V^{(1)i}\\
\nonumber&+&\frac{\sigma_1^2}{2} \sum_{l=1}^{N_1}\sum_{h=1}^{N_1}\left((\rho^2+(1-\rho^2)\rho_{1 }^2+\delta_{(1)l,(1)h}(1-\rho^2)(1-\rho_{1 }^2)\right)    \partial_{x^{(1)l}x^{(1)h}}V^{(1)i}\\
\nonumber&+&\frac{\sigma_2\sigma_1}{2}  \sum_{l=1}^{N_1}\sum_{h=1}^{N_2}\rho ^2    \partial_{x^{(1)l}x^{(2)h}}V^{(1)i}\\
\nonumber&+&\frac{\sigma_2\sigma_1}{2} \sum_{l=1}^{N_2}\sum_{h=1}^{N_1}\rho ^2    \partial_{x^{(2)l}x^{(1)h}}V^{(1)i}\\
\nonumber &+& \frac{\sigma_2^2}{2} \sum_{l=1}^{N_2}\sum_{h=1}^{N_2}  \left((\rho^2+(1-\rho^2)\rho_{2}^2+\delta_{(2)l,(2)h}(1-\rho^2)(1-\rho_{2 }^2)\right) \partial_{x^{(2)l}x^{(2)h}}V^{(1)i}\\
&+&\frac{(\alpha^{(1)i})^2}{2}-q_1\alpha^{(1)i}\left(\overline x^{\lambda_1}-x^{(1)i}\right)+\frac{\epsilon_1}{2}(\overline x^{\lambda_1}-x^{(1)i})^2\bigg\}=0,   
\ea
with the terminal condition $V^{(1)i}(T,x)=\frac{c_1}{2}(\overline x^{\lambda_1} -x^{(1)i})^2$ and 
\ba\label{HJB2}
\nonumber  \partial_{t}V^{(2)j}&+&\inf_{\alpha^{(2)j}}\bigg\{
\sum_{l=1}^{N_1}\bigg(\gamma_t^{(1)}+{\hat\alpha^{(1)l}(t,x)}\bigg)\partial_{x^{(1)l}}V^{(2)j}\\
\nonumber &+&\sum_{h\neq j,h=1}^{N_2}\bigg(\gamma_t^{(2)}+{\hat\alpha^{(2)h}(t,x)}\bigg)\partial_{x^{(2)h}}V^{(2)j}+\bigg(\gamma_t^{(2)}+{\alpha^{(2)j}}\bigg)\partial_{x^{(2)j}}V^{(2)j}\\
\nonumber&+&\frac{\sigma_1^2}{2} \sum_{l=1}^{N_1}\sum_{h=1}^{N_1}\left((\rho^2+(1-\rho^2)\rho_{1 }^2+\delta_{(1)l,(1)h}(1-\rho^2)(1-\rho_{1 }^2)\right)  \partial_{x^{(1)l}x^{(1)h}}V^{(2)j}\\
\nonumber&+&\frac{\sigma_2\sigma_1}{2}  \sum_{l=1}^{N_1}\sum_{h=1}^{N_2}\rho ^2    \partial_{x^{(1)l}x^{(2)h}}V^{(2)j}\\
\nonumber&+&\frac{\sigma_2\sigma_1}{2} \sum_{l=1}^{N_2}\sum_{h=1}^{N_1}\rho ^2    \partial_{x^{(2)l}x^{(1)h}}V^{(2)j}\\
\nonumber &+& \frac{\sigma_2^2}{2} \sum_{l=1}^{N_2}\sum_{h=1}^{N_2}  \left((\rho^2+(1-\rho^2)\rho_{2}^2+\delta_{(2)l,(2)h}(1-\rho^2)(1-\rho_{2 }^2)\right)\partial_{x^{(2)l}x^{(2)h}}V^{(2)j}\\
&+&\frac{(\alpha^{(2)j})^2}{2}-q_2\alpha^{(2)j}\left(\overline x^{\lambda_2}-x^{(2)j}\right)+\frac{\epsilon_2}{2}(\overline x^{\lambda_2}-x^{(2)j})^2\bigg\} =0,\label{HJB2}
\ea
with the terminal condition $V^{(2)j}(T,x)=\frac{c_2}{2}(\overline x^{\lambda_2} -x^{(2)j})^2$. The first order condition gives the candidate of the optimal strategy for bank $(k)i$ written as 
\be\label{candidate}
\hat\alpha^{(k)i}=q_k(\overline x^{\lambda_k}-x^{(k)i})-\partial_{x^{(k)i}}V^{(k)i}.
\en
Inserting \eqref{candidate} into \eqref{HJB1} and \eqref{HJB2} gives 
\ba 
\nonumber  \partial_{t}V^{(1)i}(t,x)&+& \bigg\{
\sum_{l=1}^{N_1}\bigg(\gamma^{(1)}_t+q_1(\overline x^{\lambda_1} -x^{(1)l})-\partial_{x^{(1)l}}V^{(1)l}\bigg)\partial_{x^{(1)l}}V^{(1)i}\\
\nonumber&+&\sum_{h=1}^{N_2}\bigg(\gamma_t^{(2)}+q_2(\overline x^{\lambda_2}-x^{(2)h})-\partial_{x^{(2)h}}V^{(2)h}\bigg)\partial_{x^{(2)h}}V^{(1)i}\\
\nonumber&+&\frac{\sigma_1^2}{2} \sum_{l=1}^{N_1}\sum_{h=1}^{N_1}\left((\rho^2+(1-\rho^2)\rho_{1 }^2+\delta_{(1)l,(1)h}(1-\rho^2)(1-\rho_{1 }^2)\right)    \partial_{x^{(1)l}x^{(1)h}}V^{(1)i}\\
\nonumber&+&\frac{\sigma_2\sigma_1}{2}  \sum_{l=1}^{N_1}\sum_{h=1}^{N_2}\rho ^2    \partial_{x^{(1)l}x^{(2)h}}V^{(1)i}\\
\nonumber&+&\frac{\sigma_2\sigma_1}{2} \sum_{l=1}^{N_2}\sum_{h=1}^{N_1}\rho ^2    \partial_{x^{(2)l}x^{(1)h}}V^{(1)i}\\
\nonumber &+& \frac{\sigma_2^2}{2} \sum_{l=1}^{N_2}\sum_{h=1}^{N_2}  \left((\rho^2+(1-\rho^2)\rho_{2}^2+\delta_{(2)l,(2)h}(1-\rho^2)(1-\rho_{2 }^2)\right) \partial_{x^{(2)l}x^{(2)h}}V^{(1)i}\\
&+&\frac{(\partial_{x^{(1)i}}V^{(1)i})^2}{2}+\frac{\epsilon_1-q_1^2}{2}(\overline x^{\lambda_1} -x^{(1)i})^2\bigg\}=0,\label{HJB1-1}
\ea
with the terminal condition $V^{(1)i}(T,x)=\frac{c_1}{2}(\overline x^{\lambda_1}  -x^{(1)i})^2$ and 
\ba 
\nonumber  \partial_{t}V^{(2)j}(t,x)&+& \bigg\{
\sum_{l=1}^{N_1}\bigg(\gamma^{(1)}_t+q_1(\overline x^{\lambda_1} -x^{(1)l})-\partial_{x^{(1)l}}V^{(1)l}\bigg)\partial_{x^{(1)l}}V^{(2)j}\\
\nonumber&+&\sum_{h=1}^{N_2}\bigg(\gamma_t^{(2)}+q_2(\overline x^{\lambda_2}-x^{(2)h})-\partial_{x^{(2)h}}V^{(2)h}\bigg)\partial_{x^{(2)h}}V^{(2)j}\\
\nonumber&+&\frac{\sigma_1^2}{2} \sum_{l=1}^{N_1}\sum_{h=1}^{N_1}\left((\rho^2+(1-\rho^2)\rho_{1 }^2+\delta_{(1)l,(1)h}(1-\rho^2)(1-\rho_{1 }^2)\right)  \partial_{x^{(1)l}x^{(1)h}}V^{(2)j}\\
\nonumber&+&\frac{\sigma_2\sigma_1}{2}  \sum_{l=1}^{N_1}\sum_{h=1}^{N_2}\rho ^2    \partial_{x^{(1)l}x^{(2)h}}V^{(2)j}\\
\nonumber&+&\frac{\sigma_2\sigma_1}{2} \sum_{l=1}^{N_2}\sum_{h=1}^{N_1}\rho ^2    \partial_{x^{(2)l}x^{(1)h}}V^{(2)j}\\
\nonumber &+& \frac{\sigma_2^2}{2} \sum_{l=1}^{N_2}\sum_{h=1}^{N_2}  \left((\rho^2+(1-\rho^2)\rho_{2}^2+\delta_{(2)l,(2)h}(1-\rho^2)(1-\rho_{2 }^2)\right)\partial_{x^{(2)l}x^{(2)h}}V^{(2)j}\\
&+&\frac{(\partial_{x^{(2)j}}V^{(2)j})^2}{2}+\frac{\epsilon_2-q_2^2}{2}(\overline x^{\lambda_2}-x^{(2)j})^2\bigg\}=0,\label{HJB2-1}
\ea
with the terminal condition $V^{(2)j}(T,x)=\frac{c_2}{2}(\overline x^{\lambda_2}-x^{(2)j})^2$. We make the ansatz for $V^{(1)i}$ written as
\ba
\nonumber V^{(1)i}(t,x)&=&\frac{\eta^{(1)}_t}{2}(\overline x^{(1)}-x^{(1)i})^2+\frac{\eta^{(2)}_t}{2}(\overline x^{(1)})^2+\frac{\eta^{(3)}_t}{2}(\overline x^{(2)})^2\\
 \nonumber&&+\eta^{(4)}_t(\overline x^{(1)}-x^{(1)i})\overline x^{(1)}+\eta^{(5)}_t(\overline x^{(1)}-x^{(1)i})\overline x^{(2)}+\eta^{(6)}_t\overline x^{(1)}\overline x^{(2)}\\
&&+\eta^{(7)}_t(\overline x^{(1)}-x^{(1)i})+\eta^{(8)}_t\overline x^{(1)}+\eta^{(9)}_t\overline x^{(2)}+\eta^{(10)}_t\label{ansatz-1},
\ea
and the ansatz for $V^{(2)j}$  given by
\ba
\nonumber V^{(2)j}(t,x)&=&\frac{\phi^{(1)}_t}{2}(\overline x^{(2)}-x^{(2)j})^2+\frac{\phi^{(2)}_t}{2}(\overline x^{(1)})^2+\frac{\phi^{(3)}_t}{2}(\overline x^{(2)})^2\\
\nonumber &&+\phi^{(4)}_t(\overline x^{(2)}-x^{(2)j})\overline x^{(1)}+\phi^{(5)}_t(\overline x^{(2)}-x^{(2)j})\overline x^{(2)}+\phi^{(6)}_t\overline x^{(1)}\overline x^{(2)}\\ 
 &&+\phi^{(7)}_t(\overline x^{(2)}-x^{(2)j})+\phi^{(8)}_t\overline x^{(1)}+\phi^{(9)}_t\overline x^{(2)}+\phi^{(10)}_t,\label{ansatz-2}
\ea
where $\eta^{(i)}$ and $\phi^{(j)}$ for $i=1,\cdots,10$ and $j=1,\cdots,10$ are deterministic functions with terminal conditions 
\ban
&&\eta_T^{(1)}=c_1,\quad  \eta_T^{(2)}=c_1\lambda_1^2(\beta_1-1)^2, \quad \eta_T^{(3)}=c_1\lambda_1^2\beta_2^2, \quad \eta_T^{(4)}=c_1\lambda_1(\beta_1-1),\\
&&\eta_T^{(5)}=c_1\lambda_1\beta_2,\quad \eta_T^{(6)}=c_1\lambda_1^2(\beta_1-1)\beta_2,\quad \eta_T^{(7)}=\eta_T^{(8)}=\eta_T^{(9)}=\eta_T^{(10)}=0,
\ean
and 
\ban
&&\phi_T^{(1)}=c_2,\quad  \phi_T^{(2)}=c_2\lambda_2^2\beta_1^2, \quad \phi_T^{(3)}=c_2\lambda_2^2(\beta_2-1)^2, \quad \phi_T^{(4)}=c_2\lambda_2\beta_1,\\
&&\phi_T^{(5)}=c_2\lambda_2(\beta_2-1),\quad \phi_T^{(6)}=c_2\lambda_2^2\beta_1(\beta_2-1),\quad \phi_T^{(7)}= \phi_T^{(8)}=\phi_T^{(9)}=\phi_T^{(10)}=0,
\ean
using 
\ban
\nonumber(\overline x^{\lambda_1} -x^{(1)i})&=& (\overline x^{(1)}-x^{(1)i})+\lambda_1(\beta_1-1)\overline x^{(1)}+\lambda_1\beta_2\overline x^{(2)}\\
 \nonumber(\overline x^{\lambda_2}-x^{(2)j})&=&(\overline x^{(2)}-x^{(2)j})+\lambda_2\beta_1\overline x^{(1)}+\lambda_2(\beta_2-1)\overline x^{(2)}.
\ean 
Inserting the ansatz \eqref{ansatz-1} and \eqref{ansatz-2} into the HJB equations \eqref{HJB1} and \eqref{HJB2} and using 
\ban
\partial_{x^{(1)l}}V^{(1)i}&=&\left(\frac{1}{N_1}-\delta_{(1)i,(1)l}\right)\eta^{(1)}_t(\overline x^{(1)}-x^{(1)i})+\frac{1}{N_1}\eta^{(2)}_t\overline x^{(1)}+\frac{1}{N_1}\eta^{(4)}_t(\overline x^{(1)}-x^{(1)i})\\
&&+\left(\frac{1}{N_1}-\delta_{(1)i,(1)l}\right)\eta_t^{(4)}\overline x^{(1)}+  \left(\frac{1}{N_1}-\delta_{(1)i,(1)l}\right)\eta_t^{(5)} \overline x^{(2)}+\frac{1}{N_1}\eta^{(6)}_t\overline x^{(2)}\\
&&+\left(\frac{1}{N_1}-\delta_{(1)i,(1)l}\right)\eta_t^{(7)}+\frac{1}{N_1}\eta_t^{(8)}\\
\partial_{x^{(2)l}}V^{(1)i}&=& \frac{1}{N_2}\eta^{(3)}_t\overline x^{(2)}+\frac{1}{N_2}\eta^{(5)}_t(\overline x^{(1)}-x^{(1)i})+\frac{1}{N_2}\eta^{(6)}_t\overline x^{(1)}+\frac{1}{N_2}\eta_t^{(9)} \\
 \ean
\ban
\partial_{x^{(1)l}x^{(1)h}}V^{(1)i}&=&\left(\frac{1}{N_1}-\delta_{(1)i,(1)l}\right)\left(\frac{1}{N_1}-\delta_{(1)i,(1)h}\right)\eta^{(1)}_t+\left(\frac{1}{N_1}\right)^2\eta_t^{(2)}\\
&&+\frac{1}{N_1}\left(\left(\frac{1}{N_1}-\delta_{(1)i,(1)h}\right)+\left(\frac{1}{N_1}-\delta_{(1)i,(1)l}\right)\right)\eta^{(4)}_t,\\
\partial_{x^{(1)l}x^{(2)h}}V^{(1)i}&=& \frac{1}{N_2}\left(\frac{1}{N_1}-\delta_{(1)i,(1)l}\right)\eta^{(5)}_t+\frac{1}{N_1N_2}\eta^{(6)}_t,\\
\partial_{x^{(2)l}x^{(1)h}}V^{(1)i}&=&\frac{1}{N_2}\left(\frac{1}{N_1}-\delta_{(1)i,(1)h}\right)\eta^{(5)}_t+\frac{1}{N_1N_2}\eta^{(6)}_t,\\
\partial_{x^{(2)l}x^{(2)h}}V^{(1)i}&=& \left(\frac{1}{N_2}\right)^2\eta_t^{(3)} ,
\ean
 
\ban
\partial_{x^{(1)l}}V^{(2)j}&=& \frac{1}{N_1}\phi^{(2)}_t\overline x^{(1)} +\frac{1}{N_1}\phi^{(4)}_t(\overline x^{(2)}-x^{(2)j})+\frac{1}{N_1}\phi^{(6)}_t\overline x^{(2)}+\frac{1}{N_1}\phi_t^{(8)}\\
\partial_{x^{(2)l}}V^{(2)j}&=&\left(\frac{1}{N_2}-\delta_{(2)j,(2)l}\right)\phi^{(1)}_t(\overline x^{(2)}-x^{(2)j})+\frac{1}{N_2}\phi^{(3)}_t\overline x^{(2)} + \left(\frac{1}{N_2}-\delta_{(2)j,(2)l}\right)\phi_t^{(4)} \overline x^{(1)}\\
&&+\frac{1}{N_2}\phi^{(5)}_t(\overline x^{(2)}-x^{(2)j})+\left(\frac{1}{N_2}-\delta_{(2)j,(2)l}\right)\phi_t^{(5)}\overline x^{(2)}+\frac{1}{N_2}\phi^{(6)}_t\overline x^{(1)},\\
&&+\left(\frac{1}{N_2}-\delta_{(2)j,(2)l}\right)\phi_t^{(7)}+\frac{1}{N_2}\phi_t^{(9)} 
\ean
\ban
\partial_{x^{(1)l}x^{(1)h}}V^{(2)j}&=&\left(\frac{1}{N_1}\right)^2\phi_t^{(2)},\\
\partial_{x^{(1)l}x^{(2)h}}V^{(2)j}&=& \frac{1}{N_1}\left(\frac{1}{N_2}-\delta_{(2)j,(2)h}\right)\phi_t^{(4)}+\frac{1}{N_1N_2}\phi_t^{(6)},\\
\partial_{x^{(2)l}x^{(1)h}}V^{(2)j}&=& \frac{1}{N_1}\left(\frac{1}{N_2}-\delta_{(2)j,(2)l}\right)\phi_t^{(4)}+\frac{1}{N_1N_2}\phi_t^{(6)},\\
\partial_{x^{(2)l}x^{(2)h}}V^{(2)j}&=&\left(\frac{1}{N_2}-\delta_{(2)j,(2)l}\right)\left(\frac{1}{N_2}-\delta_{(2)j,(2)h}\right)\phi_t^{(1)}+\left(\frac{1}{N_2}\right)^2\phi_t^{(3)}\\
&&+\frac{1}{N_2}\left(\left(\frac{1}{N_2}-\delta_{(2)j,(2)l}\right)+\left(\frac{1}{N_2}-\delta_{(2)j,(2)h}\right)\right)\phi_t^{(5)},
\ean
 we get 
\ban
\partial_tV^{1(i)}&+&\sum_{l=1}^{N_1}\bigg\{\left(q_1+(1-\frac{1}{N_1})\eta^{(1)}_t-\frac{1}{N_1}\eta^{(4)}_t\right)(\overline x^{(1)}-x^{(1)l})\\
&&-\left(\frac{1}{N_1}\eta^{(2)}_t+(\frac{1}{N_1}-1)\eta^{(4)}_t+q_1\lambda_1(1-\beta_1)\right)\overline x^{(1)}\\
&&- \left((\frac{1}{N_1}-1)\eta^{(5)}_t+\frac{1}{N_1}\eta^{(6)}_t-q_1\lambda_1\beta_2\right)\overline x^{(2)}-\left((\frac{1}{N_1}-1)\eta_t^{(7)}+\frac{1}{N_1}\eta_t^{(8)}\right)+\gamma^{(1)}_t\bigg\} \\
&&\bigg\{\left(\frac{1}{N_1}-\delta_{(1)i,(1)l}\right)\eta^{(1)}_t(\overline x^{(1)}-x^{(1)i})+\frac{1}{N_1}\eta^{(2)}_t\overline x^{(1)}+\frac{1}{N_1}\eta^{(4)}_t(\overline x^{(1)}-x^{(1)i})\\
&&+\left(\frac{1}{N_1}-\delta_{(1)i,(1)l}\right)\eta_t^{(4)}\overline x^{(1)}+  \left(\frac{1}{N_1}-\delta_{(1)i,(1)l}\right)\eta_t^{(5)} \overline x^{(2)}+\frac{1}{N_1}\eta^{(6)}_t\overline x^{(2)}\\
&&+\left(\frac{1}{N_1}-\delta_{(1)i,(1)l}\right)\eta_t^{(7)}+\frac{1}{N_1}\eta_t^{(8)}\bigg\}\\
&+&\sum_{l=1}^{N_2}\bigg\{\bigg(q_2+(1- \frac{1}{N_2})\phi^{(1)}_t-\frac{1}{N_2}\phi^{(5)}_t\bigg)(\overline x^{(2)}-x^{(2)l})\\
&&-\left((\frac{1}{N_2}-1)\phi_t^{(4)}+\frac{1}{N_2}\phi_t^{(6)}-q_2\lambda_2\beta_1\right)\overline x^{(1)}\\
&&-\left((\frac{1}{N_2}-1)\phi_t^{(5)}+\frac{1}{N_2}\phi_t^{(3)}+q_2\lambda_2(1-\beta_2)\right)\overline x^{(2)}-\frac{1}{N_2}\eta_t^{(9)} +\gamma^{(2)}_t\bigg\}\\
&& \bigg\{\frac{1}{N_2}\eta^{(3)}_t\overline x^{(2)}+\frac{1}{N_2}\eta^{(5)}_t(\overline x^{(1)}-x^{(1)i})+\frac{1}{N_2}\eta^{(6)}_t\overline x^{(1)}+\frac{1}{N_2}\eta_t^{(9)}  \bigg\}\\
\nonumber&+&\frac{\sigma_1^2}{2} \sum_{l=1}^{N_1}\sum_{h=1}^{N_1}\bigg\{\rho^2+(1-\rho^2)\rho_1^2+\delta_{(1)l,(1)h}(1-\rho^2)(1-\rho_1^2)\bigg\}\\
&&\bigg\{\left(\frac{1}{N_1}-\delta_{(1)i,(1)l}\right)\left(\frac{1}{N_1}-\delta_{(1)i,(1)h}\right)\eta^{(1)}_t+\left(\frac{1}{N_1}\right)^2\eta_t^{(2)}\\
&&+ \frac{1}{N_1}\left(\left(\frac{1}{N_1}-\delta_{(1)i,(1)h}\right)+\left(\frac{1}{N_1}-\delta_{(1)i,(1)l}\right)\right)\eta^{(4)}_t\bigg\}   \\
\nonumber&+&\frac{\rho^2\sigma_1\sigma_2}{2} \sum_{l=1}^{N_1}\sum_{h=1}^{N_2}\bigg\{\frac{1}{N_2}\left(\frac{1}{N_1}-\delta_{(1)i,(1)l}\right)\eta^{(5)}_t+\frac{1}{N_1N_2}\eta^{(6)}_t\bigg\}\\
\nonumber&+&\frac{\rho^2\sigma_2\sigma_1}{2} \sum_{l=1}^{N_2}\sum_{h=1}^{N_1} \bigg\{\frac{1}{N_2}\left(\frac{1}{N_1}-\delta_{(1)i,(1)h}\right)\eta^{(5)}_t+\frac{1}{N_1N_2}\eta^{(6)}_t\bigg\}\\
\nonumber &+& \frac{\sigma_2^2}{2} \sum_{l=1}^{N_2}\sum_{h=1}^{N_2}\bigg\{\rho_{2}^2+(1-\rho^2)\rho_2^2+\delta_{(2)l,(2)h}(1-\rho^2)(1-\rho_{2}^2)\bigg\}\bigg\{ \left(\frac{1}{N_2}\right)^2\eta_t^{(3)}\bigg\}\\
&+&\frac{1}{2}\bigg\{\left((\frac{1}{N_1}-1)\eta^{(1)}_t+\frac{1}{N_1}\eta^{(4)}_t\right)(\overline x^{(1)}-x^{(1)i})+\left(\frac{1}{N_1}\eta^{(2)}_t+(\frac{1}{N_1}-1)\eta^{(4)}_t\right)\overline x^{(1)}\\
&&+\left((\frac{1}{N_1}-1)\eta^{(5)}_t+\frac{1}{N_1}\eta^{(6)}_t\right)\overline x^{(2)}+(\frac{1}{N_1}-1)\eta_t^{(7)}+\frac{1}{N_1}\eta_t^{(8)}\bigg\}^2\\
&+&\frac{\epsilon_1-q_1^2}{2}\bigg\{ (\overline x^{(1)}-x^{(1)i})+\lambda_1(\beta_1-1)\overline x^{(1)}+\lambda_1\beta_2\overline x^{(2)}\bigg\}^2=0,
\ean
for $i=1,\cdots,N_1$ and
\ban
\partial_tV^{2(j)}&+&\sum_{l=1}^{N_1}\bigg\{\bigg(q_1+(1-\frac{1}{N_1})\eta^{(1)}_t-\frac{1}{N_1}\eta^{(4)}_t\bigg)(\overline x^{(1)}-x^{(1)l})\\
&&-\left(\frac{1}{N_1}\eta^{(2)}_t+(\frac{1}{N_1}-1)\eta^{(4)}_t+q_1\lambda_1(1-\beta_1)\right)\overline x^{(1)}\\
&&-\left(\frac{1}{N_1}\eta^{(6)}_t+(\frac{1}{N_1}-1)\eta^{(5)}_t-q_1\lambda_1\beta_2\right)\overline x^{(2)}-\frac{1}{N_1}\eta_t^{(8)}+\gamma^{(1)}_t\bigg\}\\
&& \bigg\{ \frac{1}{N_1}\phi^{(2)}_t\overline x^{(1)} +\frac{1}{N_1}\phi^{(4)}_t(\overline x^{(2)}-x^{(2)j})+\frac{1}{N_1}\phi^{(6)}_t\overline x^{(2)}+\frac{1}{N_1}\eta_t^{(8)}\bigg\}\\
&+&\sum_{l=1}^{N_2}\bigg\{\left(q_2+(1-\frac{1}{N_2})\phi^{(1)}_t-\frac{1}{N_2}\phi^{(5)}_t\right)( \overline x^{(2)}-x^{(2)l})\\
&&-\left(\frac{1}{N_2}\eta^{(6)}_t+(\frac{1}{N_2} -1)\phi^{(4)}_t-q_2\lambda_2\beta_1\right)\overline x^{(1)}\\
&&-\left((\frac{1}{N_2}-1)\phi^{(5)}_t+\frac{1}{N_2}\phi^{(3)}_t+q_2\lambda_2(1-\beta_2)\right)\overline x^{(2)}\\
&&-\left((\frac{1}{N_2}-1)\phi_t^{(7)}+\frac{1}{N_2}\phi_t^{(9)}\right)+\gamma^{(2)}_t\bigg\}\\
&& \bigg\{\left(\frac{1}{N_2}-\delta_{(2)j,(2)l}\right)\phi^{(1)}_t(\overline x^{(2)}-x^{(2)j})+\frac{1}{N_2}\phi^{(3)}_t\overline x^{(2)} + \left(\frac{1}{N_2}-\delta_{(2)j,(2)l}\right)\phi_t^{(4)} \overline x^{(1)}\\
&&+\frac{1}{N_2}\phi^{(5)}_t(\overline x^{(2)}-x^{(2)j})+\left(\frac{1}{N_2}-\delta_{(2)j,(2)l}\right)\phi_t^{(5)}\overline x^{(2)}+\frac{1}{N_2}\phi^{(6)}_t\overline x^{(1)}\\
&&+\left((\frac{1}{N_2}-\delta_{(2)j,(2)l})\phi_t^{(7)}+\frac{1}{N_2}\phi_t^{(9)}\right)\bigg\}\\
\nonumber&+&\frac{\sigma_1^2}{2} \sum_{l=1}^{N_1}\sum_{h=1}^{N_1}\bigg\{\rho^2+(1-\rho^2)\rho_1^2+\delta_{(1)l,(1)h}(1-\rho^2)(1-\rho_{1}^2)\bigg\}\bigg\{\left(\frac{1}{N_1}\right)^2\phi_t^{(2)}\bigg\}\\
\nonumber&+&\frac{\rho^2\sigma_1\sigma_2}{2} \sum_{l=1}^{N_1}\sum_{h=1}^{N_2} \bigg\{\frac{1}{N_1}\left(\frac{1}{N_2}-\delta_{(2)j,(2)h}\right)\phi_t^{(4)}+\frac{1}{N_1N_2}\phi_t^{(6)}\bigg\}\\  
\nonumber&+&\frac{\rho^2\sigma_2\sigma_1}{2} \sum_{l=1}^{N_2}\sum_{h=1}^{N_1}  \bigg\{\frac{1}{N_1}\left(\frac{1}{N_2}-\delta_{(2)j,(2)l}\right)\phi_t^{(4)}+\frac{1}{N_1N_2}\phi_t^{(6)}\bigg\}\\   
\nonumber &+& \frac{\sigma_2^2}{2} \sum_{l=1}^{N_2}\sum_{h=1}^{N_2}  \bigg\{\rho^2+(1-\rho^2)\rho_2^2+\delta_{(2)l,(2)h}(1-\rho^2)(1-\rho_{2}^2)\bigg\}\\
&&\bigg\{\left(\frac{1}{N_2}-\delta_{(2)j,(2)l}\right)\left(\frac{1}{N_2}-\delta_{(2)j,(2)h}\right)\phi_t^{(1)}+\left(\frac{1}{N_2}\right)^2\phi_t^{(3)}\\
&&+\frac{1}{N_2}\left(\left(\frac{1}{N_2}-\delta_{(2)j,(2)l}\right)+\left(\frac{1}{N_2}-\delta_{(2)j,(2)h}\right)\right)\phi_t^{(5)}\bigg\}\\
&+& \frac{1}{2}\bigg\{\left((\frac{1}{N_2}-1)\phi^{(1)}_t+\frac{1}{N_2}\phi^{(5)}_t\right)(\overline x^{(2)} -x^{(2)j})+\left(\frac{1}{N_2}\eta^{(6)}_t+(\frac{1}{N_2}-1)\phi^{(4)}_t\right)\overline x^{(1)}\\
&&+\left((\frac{1}{N_1}-1)\phi^{(5)}_t+\frac{1}{N_2}\phi^{(3)}_t\right)\overline x^{(2)}+(\frac{1}{N_2}-1)\phi_t^{(7)}+\frac{1}{N_2}\phi_t^{(9)}\bigg\}^2\\
&+&\frac{\epsilon_2-q_2^2}{2}\bigg\{(\overline x^{(2)}-x^{(2)j})+\lambda_2\beta_1\overline x^{(1)}+\lambda_2(\beta_2-1)\overline x^{(2)}\bigg\}^2=0,
\ean
for $j=1,\cdots,N_2$.
By idenfifying the terms of states $X$, we obtain that the deterministic functions $\eta^i$ and $\phi^i$ for $i=1,\cdots,10$ must satisfy
\ba
\nonumber \dot\eta^{(1)}_t&=&2\left(q_1+(1-\frac{1}{N_1} )\eta^{(1)}_t-\frac{1}{N_1}\eta^{(4)}_t\right)\eta^{(1)}_t- \left(\frac{1}{N_1}\eta^{(4)}_t+(\frac{1}{N_1}-1)\eta_t^{(1)}\right)^2-(\epsilon_1-q_1^2)\\
\label{eta1}\\
 \nonumber  {\dot\eta^{(2)}_t}&=&2\left(\frac{1}{N_1}\eta^{(2)}_t+(\frac{1}{N_1}  -1)\eta^{(4)}_t+q_1\lambda_1(1-\beta_1)\right)\eta^{(2)}_t- \left(\frac{1}{N_1}\eta^{(2)}_t+(\frac{1}{N_1} -1)\eta_t^{(4)}\right)^2\\
\label{eta2}&&+2\left(\frac{1}{N_2}\phi^{(6)}_t+(\frac{1}{N_2}-1)\phi^{(4)}_t-q_2\lambda_2\beta_1\right)\eta^{(6)}_t-(\epsilon_1-q_1^2)\lambda_1^2(\beta_1-1)^2\\
\nonumber  {\dot\eta^{(3)}_t} &=&2\left(\frac{1}{N_2}\phi^{(3)}_t+(\frac{1}{N_2}-1)\phi^{(5)}_t+q_2\lambda_2(1-\beta_2)\right)\eta^{(3)}_t-\left(\frac{1}{N_1}\eta_t^{(6)}+(\frac{1}{N_1}-1)\eta_t^{(5)}\right)^2\\
&&+2\left(\frac{1}{N_1}\eta^{(6)}_t+(\frac{1}{N_1}-1)\eta^{(5)}_t-q_1\lambda_1\beta_2\right)\eta^{(6)}_t-(\epsilon_1-q_1^2)\lambda_1^2\beta_2^2
\label{eta3}
\ea
\ba
\nonumber\dot\eta^{(4)}_t&=&\left(q_1+(1-\frac{1}{N_1})\eta^{(1)}_t-\frac{1}{N_1}\eta^{(4)}_t\right) \eta^{(4)}_t+\left(\frac{1}{N_1}\eta^{(2)}_t+(\frac{1}{N_1}-1)\eta^{(4)}_t+q_1\lambda_1(1-\beta_1)\right)\eta^{(4)}_t\\
\nonumber&&-\left(\frac{1}{N_1}\eta^{(4)}_t+(\frac{1}{N_1}-1)\eta_t^{(1)}\right)\left(\frac{1}{N_1}\eta^{(2)}_t+(\frac{1}{N_1}-1)\eta_t^{(4)}\right)\\
\label{eta4}&&+\left(\frac{1}{N_2}\phi^{(6)}_t+(\frac{1}{N_2}-1)\phi^{(4)}_t-q_2\lambda_2\beta_1\right)\eta^{(5)}_t-(\epsilon_1-q_1^2)\lambda_1(\beta_1-1)\\
\nonumber\dot\eta^{(5)}_t&=&\left(q_1+(1-\frac{1}{N_1})\eta^{(1)}_t-\frac{1}{N_1}\eta^{(4)}_t\right) \eta^{(5)}_t+\left(\frac{1}{N_2}\phi^{(3)}_t+(\frac{1}{N_2}-1)\phi^{(5)}_t+q_2\lambda_2(1-\beta_2)\right)\eta^{(5)}_t\\
\nonumber&&+\left(\frac{1}{N_1}\eta^{(6)}_t+(\frac{1}{N_1}-1)\eta^{(5)}_t-q_1\lambda_1\beta_2\right)\eta^{(4)}_t\\
\label{eta5}&&-\left(\frac{1}{N_1}\eta^{(4)}_t+(\frac{1}{N_1}-1)\eta_t^{(1)}\right)\left(\frac{1}{N_1}\eta^{(6)}_t+(\frac{1}{N_1}-1)\eta_t^{(5)}\right)-(\epsilon_1-q_1^2)\lambda_1\beta_2\\
\nonumber\dot\eta^{(6)}_t&=&\left(\frac{1}{N_1}\eta^{(2)}_t+(\frac{1}{N_1}-1)\eta^{(4)}_t+q_1\lambda_1(1-\beta_1)\right)\eta^{(6)}_t\\
\nonumber&&+\left(\frac{1}{N_2}\phi^{(3)}_t+(\frac{1}{N_2}-1)\phi^{(5)}_t+q_2\lambda_2(1-\beta_2)\right) \eta^{(6)}_t\\
\nonumber&&+\left(\frac{1}{N_1}\eta^{(6)}_t+(\frac{1}{N_1}-1)\eta^{(5)}_t-q_1\lambda_1\beta_2\right)\eta^{(2)}_t\\
\nonumber &&-\left(\frac{1}{N_1}\eta^{(2)}_t+(\frac{1}{N_1}-1)\eta_t^{(4)}\right)\left(\frac{1}{N_1}\eta^{(6)}_t+(\frac{1}{N_1}-1)\eta_t^{(5)}\right)\\
 &&+\left(\frac{1}{N_2}\phi^{(6)}_t+(\frac{1}{N_2}-1)\phi^{(4)}_t-q_2\lambda_2\beta_1\right) \eta^{(3)}_t-(\epsilon_1-q_1^2)\lambda_1^2(\beta_1-1)\beta_2\label{eta6}
 \ea
 \ba
 \nonumber \dot\eta^{(7)}_t&=&\left(q_1+(1-\frac{1}{N_1})\eta_t^{(1)}-\frac{1}{N_1}\eta_t^{(4)}\right)\eta_t^{(7)}\\
 \nonumber &&+\left((\frac{1}{N_1}-1)\eta_t^{(7)}+\frac{1}{N_1}\eta_t^{(8)}-\gamma_t^{(1)}\right)\eta_t^{(4)}+\left((\frac{1}{N_2}-1)\phi_t^{(7)}+\frac{1}{N_2}\phi_t^{(9)}-\gamma_t^{(2)}\right)\eta_t^{(5)}\\
 &&-\left((\frac{1}{N_1}-1)\eta_t^{(1)}+\frac{1}{N_1}\eta_t^{(4)}\right)\left((\frac{1}{N_1}-1)\eta_t^{(7)}+\frac{1}{N_1}\eta_t^{(8)}\right)\label{eta7}\\
 \nonumber \dot\eta^{(8)}_t&=&\left(\frac{1}{N_1}\eta_t^{(2)}+(\frac{1}{N_1}-1)\eta_t^{(4)}+q_1\lambda_1(1-\beta_1)\right)\eta_t^{(8)}+\left((\frac{1}{N_1}-1)\eta_t^{(7)}+\frac{1}{N_1}\eta_t^{(8)}-\gamma_t^{(1)}\right)\eta_t^{(2)}\\
  \nonumber&&-\left(\frac{1}{N_1}\eta^{(2)}_t+(\frac{1}{N_1}-1)\eta_t^{(4)}\right)\left((\frac{1}{N_1}-1)\eta_t^{(7)}+\frac{1}{N_1}\eta_t^{(8)}\right)\\
\nonumber &&+\left(\frac{1}{N_2}\phi^{(6)}_t+(\frac{1}{N_2}-1)\phi^{(4)}_t-q_2\lambda_2\beta_1\right)\eta_t^{(9)}+\left((\frac{1}{N_2}-1)\phi_t^{(7)}+\frac{1}{N_2}\phi_t^{(9)}-\gamma_t^{(2)}\right)\eta_t^{(6)}\\
 \label{eta8}\\
  \nonumber \dot\eta^{(9)}_t&=&\left((\frac{1}{N_2}-1)\phi_t^{(5)}+\frac{1}{N_2}\phi_t^{(3)}+q_2\lambda_2(1-\beta_2)\right)\eta_t^{(9)}+\left((\frac{1}{N_1}-1)\eta_t^{(7)}+\frac{1}{N_1}\eta_t^{(8)}-\gamma_t^{(1)}\right)\eta_t^{(6)}\\
  \nonumber   &&+\left(\frac{1}{N_1}\eta^{(6)}_t+(\frac{1}{N_1}-1)\eta_t^{(5)}-q_1\lambda_1\beta_2\right)\eta_t^{(8)}+\left((\frac{1}{N_2}-1)\phi_t^{(7)}+\frac{1}{N_2}\phi_t^{(9)}-\gamma_t^{(2)}\right)\eta_t^{(3)}\\
\label{eta9}  &&-\left(\frac{1}{N_1}\eta^{(6)}_t+(\frac{1}{N_1}-1)\eta_t^{(5)}\right)\left((\frac{1}{N_1}-1)\eta_t^{(7)}+\frac{1}{N_1}\eta_t^{(8)}\right)\\
 \nonumber \dot\eta^{(10)}_t&=&\left((\frac{1}{N_1}-1)\eta_t^{(7)}+\frac{1}{N_1}\eta_t^{(8)}-\gamma^{(1)}_t\right)\eta_t^{(8)}+\left((\frac{1}{N_2}-1)\phi_t^{(7)}+\frac{1}{N_2}\phi_t^{(9)}-\gamma^{(2)}_t\right)\eta_t^{(9)}\\
 \nonumber &&-\frac{\sigma_1^2}{2}\left((1-\frac{1}{N_1})(1-\rho^2)(1-\rho_1^2)\eta_t^{(1)}+\left(\rho^2+(1-\rho^2)\rho_1^2+\frac{1}{N_1}(1-\rho^2)(1-\rho^2_1)\right)\eta_t^{(2)}\right)\\
 \nonumber &&-\rho^2\sigma_1\sigma_2\eta_t^{(6)}-\frac{\sigma_2^2}{2}\left(\rho_2^2+(1-\rho^2)\rho^2_2+\frac{1}{N_2}(1-\rho^2)(1-\rho_2^2)\right)\eta_t^{(3)}\\
 &&-\frac{1}{2}\left((\frac{1}{N_1}-1)\eta_t^{(7)}+\frac{1}{N_1}\eta_t^{(8)}\right)^2,\label{eta10}
\ea
and
\ba
\nonumber\dot\phi^{(1)}_t&=&2\left(q_2+(1-\frac{1}{N_2})\phi^{(1)}_t-\frac{1}{N_2}\phi_t^{(5)}\right)\phi^{(1)}_t-\left(( \frac{1}{N_2}-1)\phi^{(1)}_t-\frac{1}{N_2}\phi_t^{(5)}\right)^2-(\epsilon_2-q_2^2)\\
\label{phi1}\\
\nonumber {\dot\phi^{(2)}_t}&=&2\left(\frac{1}{N_1}\eta^{(2)}_t+(\frac{1}{N_1}-1)\eta^{(4)}_t+q_1\lambda_1(1-\beta_1)\right) \phi^{(2)}_t- \left(\frac{1}{N_2}\phi^{(6)}_t+(\frac{1}{N_2}-1)\phi_t^{(4)}\right)^2 \\
 &&+2\left(\frac{1}{N_2}\phi^{(6)}_t+(\frac{1}{N_2}-1)\phi^{(4)}_t-q_2\lambda_2\beta_1\right)\phi^{(6)}_t-(\epsilon_2-q_2^2)\lambda_2^2\beta_2^2 \label{phi2}\\
\nonumber {\dot\phi^{(3)}_t} &=&2\left( \frac{1}{N_2}\phi^{(3)}_t+(\frac{1}{N_2}-1)\phi^{(5)}_t+q_2\lambda_2(1-\beta_2)\right) \phi^{(3)}_t - \left(\frac{1}{N_2}\phi_t^3+(\frac{1}{N_2}-1)\phi_t^{(5)}\right)^2\\
 &&+2\left(\frac{1}{N_1}\eta^{(6)}_t+(\frac{1}{N_1}-1)\eta^{(5)}_t-q_1\lambda_1\beta_2\right) \phi^{(6)}_t- (\epsilon_2-q_2^2)\lambda_2^2(\beta_2-1)^2
\label{phi3}
\ea
\ba
\nonumber\dot\phi^{(4)}_t&=&\left(q_2+(1-\frac{1}{N_2})\phi^{(1)}_t-\frac{1}{N_2}\phi_t^{(5)}\right)\phi^{(4)}_t+\left(\frac{1}{N_2}\phi^{(6)}_t+(\frac{1}{N_2}-1)\phi^{(4)}_t-q_2\lambda_2\beta_1\right)\phi^{(5)}_t\\
\nonumber&&-\left((\frac{1}{N_2}-1)\phi^{(1)}_t+\frac{1}{N_2}\phi_t^{(5)}\right)\left(\frac{1}{N_2}\phi^{(6)}_t+(\frac{1}{N_2}-1)\phi_t^{(4)}\right)\\ 
&&+\left(\frac{1}{N_1}\eta^{(2)}_t+(\frac{1}{N_1}-1)\eta^{(4)}_t+q_1\lambda_1(1-\beta_1)\right)\phi^{(4)}_t-(\epsilon_2-q_2^2)\lambda_2\beta_1\label{phi4}\\
\nonumber\dot\phi^{(5)}_t&=&\left(q_2+(1-\frac{1}{N_2})\phi^{(1)}_t-\frac{1}{N_2}\phi_t^{(5)} \right)\phi^{(5)}_t+\left(\frac{1}{N_2}\phi^{(3)}_t+(\frac{1}{N_2}-1)\phi_t^{(5)}+q_2\lambda_2(1-\beta_2)\right)\phi^{(5)}_t\\
\nonumber&&-\left((\frac{1}{N_2}-1)\phi^{(1)}_t+\frac{1}{N_2}\phi_t^{(5)}\right)\left(\frac{1}{N_2}\phi^{(3)}_t+(\frac{1}{N_2}-1)\phi_t^{(5)}\right)\\
\label{phi5}&&+\left(\frac{1}{N_1}\eta^{(6)}_t+(\frac{1}{N_1}-1)\eta^{(5)}_t-q_1\lambda_1\beta_2\right)\phi_t^{(4)}-(\epsilon_2-q_2^2)\lambda_2(\beta_2-1)\\
\nonumber\dot\phi^{(6)}_t&=&\left( \frac{1}{N_1}\eta^{(2)}_t+(\frac{1}{N_1}-1)\eta^{(4)}_t+q_1\lambda_1(1-\beta_1)\right) \phi^{(6)}_t \\
\nonumber&&+\left(\frac{1}{N_2}\phi^{(3)}_t+(\frac{1}{N_2}-1)\phi^{(5)}_t+q_2\lambda_2(1-\beta_2)\right) \phi^{(6)}_t \\
\nonumber &&-\left(\frac{1}{N_2}\phi^{(6)}_t+(\frac{1}{N_2}-1)\phi_t^{(4)} \right) \left(\frac{1}{N_2}\phi^{(3)}_t+(\frac{1}{N_2}-1)\phi_t^{(5)}\right)\\
\nonumber&&+\left(\frac{1}{N_1}\eta^{(6)}_t+(\frac{1}{N_1}-1)\eta^{(5)}_t-q_1\lambda_1\beta_2\right) \phi^{(2)}_t +\left(\frac{1}{N_2}\phi^{(6)}_t+(\frac{1}{N_2}-1)\phi^{(4)}_t-q_2\lambda_2\beta_1\right)\phi^{(3)}_t\\
\label{phi6}&&- (\epsilon_2-q_2^2)\lambda_2^2\beta_1(\beta_2-1)\\
\nonumber\dot\phi^{(7)}_t&=&\left(q_2+(1-\frac{1}{N_2})\phi_t^{(1)}-\frac{1}{N_2}\phi_t^{(5)}\right)\phi_t^{(7)}\\
 \nonumber &&+\left((\frac{1}{N_1}-1)\eta_t^{(7)}+\frac{1}{N_1}\eta_t^{(8)}-\gamma_t^{(1)}\right)\phi_t^{(4)}+\left((\frac{1}{N_2}-1)\phi_t^{(7)}+\frac{1}{N_2}\phi_t^{(9)}-\gamma_t^{(2)}\right)\phi_t^{(5)}\\
 &&-\left((\frac{1}{N_2}-1)\phi_t^{(1)}+\frac{1}{N_2}\phi_t^{(5)}\right)\left((\frac{1}{N_2}-1)\phi_t^{(7)}+\frac{1}{N_2}\phi_t^{(9)}\right)\label{phi7}\\
  \nonumber \dot\phi^{(8)}_t&=&\left(\frac{1}{N_1}\eta_t^{(2)}+(\frac{1}{N_1}-1)\eta_t^{(4)}+q_1\lambda_1(1-\beta_1)\right)\phi_t^{(8)}+\left((\frac{1}{N_1}-1)\eta_t^{(7)}+\frac{1}{N_1}\eta_t^{(8)}-\gamma_t^{(1)}\right)\phi_t^{(2)}\\
  \nonumber&&-\left((\frac{1}{N_2}-1)\phi_t^{(4)}+\frac{1}{N_2}\phi_t^{(6)}\right)\left((\frac{1}{N_2}-1)\phi_t^{(7)}+\frac{1}{N_2}\phi_t^{(9)}\right)\\
 \nonumber&&+\left(\frac{1}{N_2}\phi^{(6)}_t+(\frac{1}{N_2}-1)\phi^{(4)}_t-q_2\lambda_2\beta_1\right)\phi_t^{(9)}+\left((\frac{1}{N_2}-1)\phi_t^{(7)}+\frac{1}{N_2}\phi_t^{(9)}-\gamma_t^{(2)}\right)\phi_t^{(6)}\\\label{phi8}\\
  \nonumber \dot\phi^{(9)}_t&=&\left((\frac{1}{N_2}-1)\phi_t^{(5)}+\frac{1}{N_2}\phi_t^{(3)}+q_2\lambda_2(1-\beta_2)\right)\phi_t^{(9)}+\left((\frac{1}{N_1}-1)\eta_t^{(7)}+\frac{1}{N_1}\eta_t^{(8)}-\gamma_t^{(1)}\right)\phi_t^{(6)}\\
  \nonumber   &&+\left(\frac{1}{N_1}\eta^{(6)}_t+(\frac{1}{N_1}-1)\eta_t^{(5)}-q_1\lambda_1\beta_2\right)\phi_t^{(8)}+\left((\frac{1}{N_2}-1)\phi_t^{(7)}+\frac{1}{N_2}\phi_t^{(9)}-\gamma_t^{(2)}\right)\phi_t^{(3)}\\
\label{phi9}  &&+\left((\frac{1}{N_2}-1)\phi_t^{(5)}+\frac{1}{N_2}\phi_t^{(3)}\right)\left((\frac{1}{N_2}-1)\phi_t^{(7)}+\frac{1}{N_2}\phi_t^{(9)}\right)
\ea
\ba
 \nonumber \dot\phi^{(10)}_t&=& \left((\frac{1}{N_1}-1)\eta_t^{(7)}+\frac{1}{N_1}\phi_t^{(8)}-\gamma^{(1)}_t\right)\eta_t^{(8)}+\left((\frac{1}{N_2}-1)\phi_t^{(7)}+\frac{1}{N_2}\phi_t^{(9)}-\gamma^{(2)}_t\right)\phi_t^{(9)}\\
  \nonumber &&-\frac{\sigma_1^2}{2}\left(\rho_1^2+(1-\rho^2)\rho_1+\frac{1}{N_1}(1-\rho^2)(1-\rho_1^2)\right)\phi_t^{(2)}-\rho^2\sigma_1\sigma_2\phi_t^{(6)}\\
  \nonumber  &&-\frac{\sigma_2^2}{2}\left((1-\frac{1}{N_2})(1-\rho^2)(1-\rho_2^2)\phi_t^{(1)}+\left(\rho^2+(1-\rho^2)\rho_2^2+\frac{1}{N_2}(1-\rho^2)(1-\rho^2_2)\right)\phi_t^{(3)}\right)\\
  &&-\frac{1}{2}\left((\frac{1}{N_1}-1)\eta_t^{(7)}+\frac{1}{N_1}\eta_t^{(8)}\right)^2,\label{phi10}
\ea
with terminal conditions  
\ban
&&\eta_T^{(1)}=c_1,\quad  \eta_T^{(2)}=c_1\lambda_1^2(\beta_1-1)^2, \quad \eta_T^{(3)}=c_1\lambda_1^2\beta_2^2, \quad \eta_T^{(4)}=c_1\lambda_1(\beta_1-1),\\
&&\eta_T^{(5)}=c_1\lambda_1\beta_2,\quad \eta_T^{(6)}=c_1\lambda_1^2(\beta_1-1)\beta_2,\quad \eta_T^{(7)}=\eta_T^{(8)}=\eta_T^{(9)}=\eta_T^{(10)}=0,
\ean
and 
\ban
&&\phi_T^{(1)}=c_2,\quad  \phi_T^{(2)}=c_2\lambda_2^2\beta_1^2, \quad \phi_T^{(3)}=c_2\lambda_2^2(\beta_2-1)^2, \quad \phi_T^{(4)}=c_2\lambda_2\beta_1,\\
&&\phi_T^{(5)}=c_2\lambda_2(\beta_2-1),\quad \phi_T^{(6)}=c_2\lambda_2^2\beta_1(\beta_2-1),\quad \phi_T^{(7)}= \phi_T^{(8)}=\phi_T^{(9)}=\phi_T^{(10)}=0.
\ean
We now discuss the existence of $\eta^{(i)}$ and $\phi^{(i)}$ for $i=1,\cdots,10$. First, observe that $\eta^{(3)}$, $\eta^{(i)}$ for $i=5,\cdots,10$, $\phi^{(2)}$, $\phi^{(4)}$, and $\phi^{(i)}$ for $i=6,\cdots,10$ are coupled first order linear equations.

The existence of $\eta^{(i)}$ and $\phi^{(i)}$ for $i=1,\cdots,10$ in the case of sufficiently large $N_1$ and $N_2$ can be verified in Proposition \ref{Prop_suff}. Hence, the closed-loop Nash equilibria are written as 
\ba
\label{optimal-finite-ansatz-V1-appen}
\hat\alpha^{(1)i}(t,x)&=&(q_1+\tilde\eta^{(1)}_t)(\overline x^{(1)}-x^{(1)i})+\tilde\eta^{(4)}_t\overline x^{(1)}+\tilde\eta^{(5)}_t\overline x^{(2)}+\tilde\eta_t^{(7)},\\
\hat\alpha^{(2)j}(t,x)&=&(q_2+\tilde\phi^{(1)}_t)(\overline x^{(2)}-x^{(2)j})+\tilde\phi^{(4)}_t\overline x^{(1)}+\tilde\phi^{(5)}_t\overline x^{(2)}+\tilde\phi_t^{(7)},\label{optimal-finite-ansatz-V2-appen}
\ea
where $\tilde\eta^i$ and $\tilde\phi^i$ for $i=1,4,5,7$ satisfy (\ref{tildeeta}-\ref{tildephi}).

\qed 


We then verify that $V^{(1)i}$, $V^{(2)j}$, $\hat\alpha^{(1)i}$, and $\hat\alpha^{(2)j}$ are the solutions to the problem (\ref{value-function-1}-\ref{coupled-2}). Without loss of generality, we show the verification theorem for $V^{(1)i}$.

{\theorem\label{Ver-Thm}(Verification Theorem)\\
Given the optimal strategies $\hat\alpha^{(1)l}$ for $l\neq i$ given by  \eqref{optimal-finite-ansatz-V1} and $\hat\alpha^{(2)j}$ for $j=1,\cdots,N_2$ given by \eqref{optimal-finite-ansatz-V2}, $V^{(1)i}$ given by \eqref{ansatz-1} is the value function associated to the problem \eqref{value-function-1} and  \eqref{value-function-2} subject to \eqref{coupled-1} and \eqref{coupled-2} and $\hat\alpha^{(1)i}$  is the optimal strategy for the $i$-th bank in the first group   and also the closed-loop Nash equilibrium. 
} 
\begin{proof}
According to the notations in \cite{Sun2016}, an admissible strategy $\tilde\alpha$ and its corresponding trajectory $\tilde X$ are given by  
\be
\tilde\alpha_t=\left(\hat\alpha_t^{(1)1},\cdots, \alpha^{(1)i}_t,\cdots,\hat\alpha_t^{(1)N_1},\hat\alpha_t^{(2)1},\cdots,\hat\alpha_t^{(2)N_2}\right)
\en
and
\be
\tilde X_t=\left(\tilde X_t^{(1)1},\cdots,\tilde X^{(1)i}_t,\cdots,\tilde X_t^{(1)N_1},\tilde X_t^{(2)1},\cdots,\tilde X_t^{(2)N_2}\right).
\en
In addition, the optimal strategy $\hat\alpha$ and its corresponding trajectory $\hat X$ are written as 
\be
\hat\alpha_t=\left(\hat\alpha_t^{(1)1},\cdots,\hat\alpha^{(1)i}_t,\cdots,\hat\alpha_t^{(1)N_1},\hat\alpha_t^{(2)1},\cdots,\hat\alpha_t^{(2)N_2}\right)
\en
and
\be
\hat X_t=\left(\hat X_t^{(1)1},\cdots,\hat X^{(1)i}_t,\cdots,\hat X_t^{(1)N_1},\hat X_t^{(2)1},\cdots,\hat X_t^{(2)N_2}\right).
\en
We claim for any admissible strategy $\tilde\alpha$, 
\be \label{upper}
V^{(1)i}(t,x)\leq \EE_{t,x}\left\{\int_t^T f^N_{(1)}(\tilde{X}_t, \alpha^{(1)i}_s)ds+g_{(1)}(\tilde{X}_T)\right\},
\en
and for $\hat{\alpha}$ 
\be \label{optim}
V^{(1)i}(t,x)= \EE_{t,x}\left\{\int_t^T f^N_{(1)}(\hat{X}_t, \hat{\alpha}^{(1)i}_s)ds+g_{(1)}(\hat{X}_T)\right\},
\en
leading to $\hat{\alpha}^{(1)i}$ is the optimal strategy for $i$-th bank in the first group .  We can assume 
\be \label{condition}
\EE_{t,x}\left\{\int_t^T f^N_{(1)}(\tilde{X}_t,  {\alpha}^{(1)i}_s)ds\right\}<\infty,
\en
otherwise (\ref{upper}) holds automatically. For some $M>0$, define the exit time 
\[
\theta_M=\inf\{t;\;  |\tilde{X}_t|\geq M \}.
\]
Given the condition \eqref{condition}, in order to complete the proof, we shall claim 
\be \label{nonexplosion}
P(\theta_M\leq T)\rightarrow 0,\ M\rightarrow \infty,
\en
and
\be \label{ui}
\EE_{t, x}[\sup_{t\leq s\leq T}|\tilde{X}_s|^2]<\infty.
\en
The proof of the above properties is postponed. 

Given the optimal strategies $\hat\alpha^{(1)i}$ and $\hat\alpha^{(2)j}$, applying It\^o's formula, we get 
\ba
\nonumber&&V^{(1)i}(T\wedge\theta_M,\tilde{X}_{T\wedge\theta_M})\\
\nonumber&=&V^{(1)i}(t,x)\\
\nonumber &&+\int_t^{T\wedge\theta_M}\bigg\{\partial_sV^{(1)i}(s,\tilde X_s)+
\sum_{l\neq i,l=1}^{N_1}\bigg(\gamma^{(1)}_t+{\hat\alpha^{(1)l}(t,x)}\bigg)\partial_{x^{(1)l}}V^{(1)i}(s,\tilde X_s)
\\
\nonumber&&+ \bigg(\gamma^{(1)}_t+{\alpha^{(1)i}}\bigg)\partial_{x^{(1)i}}V^{(1)i}(s,\tilde X_s)+\sum_{h=1}^{N_2}\bigg(\gamma_t^{(2)}+{\hat\alpha^{(2)h}(t,x)}\bigg)\partial_{x^{(2)h}}V^{(1)i}(s,\tilde X_s)\\
\nonumber&&+\frac{(\sigma^1)^2}{2} \sum_{l=1}^{N_1}\sum_{h=1}^{N_1}((\rho^{11})^2+\delta_{(1)l,(1)h}(1-(\rho^{11})^2))    \partial_{x^{(1)l}x^{(1)h}}V^{(1)i}(s,\tilde X_s)\\
\nonumber&&+\frac{\sigma^1\sigma^2}{2} \sum_{l=1}^{N_1}\sum_{h=1}^{N_2}((\rho^{12})^2+\delta_{(1)l,(2)h}(1-(\rho^{12})^2))    \partial_{x^{(1)l}x^{(2)h}}V^{(1)i}(s,\tilde X_s)\\
\nonumber&&+\frac{\sigma^2\sigma^1}{2} \sum_{l=1}^{N_2}\sum_{h=1}^{N_1}((\rho^{21})^2+\delta_{(2)l,(1)h}(1-(\rho^{21})^2))    \partial_{x^{(2)l}x^{(1)h}}V^{(1)i}(s,\tilde X_s)\\
\nonumber &&+ \frac{(\sigma^2)^2}{2} \sum_{l=1}^{N_2}\sum_{h=1}^{N_2}  ((\rho^{22})^2+\delta_{(2)l,(2)h}(1-(\rho^{22})^2))\partial_{x^{(2)l}x^{(2)h}}V^{(1)i}(s,\tilde X_s)\bigg\} ds\\
\nonumber&&+\int_t^{T\wedge\theta_M}\sigma^1\sum_{l=1}^{N_1}\partial_{x^{(1)l}}V^{(1)i}(s, \tilde{X}_s) dW^{(1)l}_s+\int_t^{T\wedge\theta_M}\sigma^2\sum_{h=1}^{N_1}\partial_{x^{(2)h}}V^{(1)i}(s, \tilde{X}_s) dW^{(2)h}_s.
\ea
Taking the expectation on both sides and using 
\ba \label{positive}
\nonumber  \partial_{t}V^{(1)i} &+& 
\sum_{l\neq i,l=1}^{N_1}\bigg(\gamma^{(1)}_t+{\hat\alpha^{(1)l}(t,x)}\bigg)\partial_{x^{(1)l}}V^{(1)i}
+ \bigg(\gamma^{(1)}_t+{\alpha^{(1)i}}\bigg)\partial_{x^{(1)i}}V^{(1)i}\\
\nonumber&+&\sum_{h=1}^{N_2}\bigg(\gamma_t^{(2)}+{\hat\alpha^{(2)h}(t,x)}\bigg)\partial_{x^{(2)h}}V^{(1)i}\\
\nonumber&+&\frac{(\sigma^1)^2}{2} \sum_{l=1}^{N_1}\sum_{h=1}^{N_1}((\rho^{11})^2+\delta_{(1)l,(1)h}(1-(\rho^{11})^2))    \partial_{x^{(1)l}x^{(1)h}}V^{(1)i}\\
\nonumber&+&\frac{\sigma^1\sigma^2}{2} \sum_{l=1}^{N_1}\sum_{h=1}^{N_2}((\rho^{12})^2+\delta_{(1)l,(2)h}(1-(\rho^{12})^2))    \partial_{x^{(1)l}x^{(2)h}}V^{(1)i}\\
\nonumber&+&\frac{\sigma^2\sigma^1}{2} \sum_{l=1}^{N_2}\sum_{h=1}^{N_1}((\rho^{21})^2+\delta_{(2)l,(1)h}(1-(\rho^{21})^2))    \partial_{x^{(2)l}x^{(1)h}}V^{(1)i}\\
\nonumber &+& \frac{(\sigma^2)^2}{2} \sum_{l=1}^{N_2}\sum_{h=1}^{N_2}  ((\rho^{22})^2+\delta_{(2)l,(2)h}(1-(\rho^{22})^2))\partial_{x^{(2)l}x^{(2)h}}V^{(1)i}\\
&+&\frac{(\alpha^{(1)i})^2}{2}-q^1\alpha^{(1)i}\left(\overline x^{\lambda_1}-x^{(1)i}\right)+\frac{\epsilon^1}{2}(\overline x^{\lambda_1}-x^{(1)i})^2  \geq 0
\ea
give
\ba
\nonumber  V^{(1)i}(t,x) &\leq&\EE\bigg\{\int_t^{T\wedge\theta_M}\left(\frac{(\alpha^{(1)i}_s)^2}{2}-q_1\alpha^{(1)i}_s\left(\overline{\tilde X}^{\lambda_1}_s-\tilde X^{(1)i}_s\right)+\frac{\epsilon_1}{2}(\overline {\tilde X}^{\lambda_1}_s-\tilde X^{(1)i}_s)^2\right)ds \\
&&+ V^{(1)i}(T\wedge\theta_M,\tilde X_{T\wedge\theta_M})\bigg\}.\label{ver-ineq}
\ea
Assuming that there exists a constant $C$ such that 
$$
|V^{(1)i}(T\wedge\theta_M,\tilde X_{T\wedge\theta_M})|\leq C(1+ \sup_{t\leq s\leq T} |\tilde{X}_s|^2),
$$ 
and then the condition \eqref{ui} implies  $V^{(1)i}(T\wedge\theta_M,X_{T\wedge\theta_M})$ being uniformly integrable. Together with (\ref{nonexplosion}), we obtain as $M\rightarrow\infty$
$$
\EE_{t, x}[ V^{(1)i}(T\wedge\theta_M,\tilde X_{T\wedge\theta_M})]\rightarrow \EE_{t, x}[g^N_{(1)}(\tilde{X}_T)].
$$
In addition, owing to the integrand staying nonnegative, we get
\ban
\nonumber && \EE\bigg\{\int_t^{T\wedge\theta_M}\left(\frac{(\alpha^{(1)i}_s)^2}{2}-q\alpha^{(1)i}_s\left(\overline X^{\lambda_1}_s-\tilde{X}^{(1)i}_s\right)+\frac{\epsilon}{2}(\overline X^{\lambda_1}_s-\tilde{X}^{(1)i}_s)^2\right)ds \bigg\}\\
&&  \leq \EE\bigg\{\int_t^{T}\left(\frac{(\alpha^{(1)i}_s)^2}{2}-q\alpha^{(1)i}_s\left(\overline X^{\lambda_1}_s-\tilde{X}^{(1)i}_s\right)+\frac{\epsilon}{2}(\overline X^{\lambda_1}_s-\tilde{X}^{(1)i}_s)^2\right)ds \bigg\}.
\ean
Based on the above results, we obtain
\ba
V^{(1)i}(t,x_t) \leq\EE\bigg\{\int_t^{T}f^N_{(1)}(\tilde{X}_t, \alpha^{(1)i}_s)ds 
+g^N_{(1)}(\tilde{X}_T))\bigg\}.\label{V(t,x)}
\ea
This completes the proof of (\ref{upper}).
In order to prove \eqref{nonexplosion}, we first recall 
\be\label{coupled-appex}
d\tilde X_t^{(1)i}=(\alpha^{(1)i}_t+\gamma_t^{(1)})dt+ \sigma_1dW^{(1)i}_t
\en
and
\[
d(\tilde X_t^{(1)i})^2=\left(2\tilde X_t^{(1)i}(\alpha^{(1)i}_t+\gamma^{(1)}_t)+\sigma_1^2\right)dt+2\tilde X_t^{(1)i}\sigma_1dW^{(1)i}_t. 
\]
Denote $\beta>0$ large enough satisfying
\[
\beta>2\sup_{0\leq t\leq T}\left\{(\gamma^{(1)}_t)^2+\sigma_1^2+1\right\}.
\] 
and apply It\^o formula to $e^{-\beta t}(\tilde X_t^{(1)i})^2$ leading to 
\ban
&&e^{-\beta(t\wedge\theta_M)}(\tilde X^{(1)i})^2_{t\wedge\theta_M}\\
&\leq&(\tilde X_0^{(1)i})^2+\int_0^{t\wedge\theta_M}e^{-\beta s}\left(-\frac{\beta}{2}((\tilde X_s^{(1)i})^2-1)+|\alpha^{(1)i}_s|^2\right)ds+\int_0^{t\wedge\theta_M}e^{-\beta s}2\sigma_1\tilde X_s^{(1)i}dW^{(1)i}_t. 
\ean
Taking expectation on both sides gives 
\be
e^{-\beta t}M^2\PP(\theta_M\leq t)\leq (\tilde X_0^{(1)i})^2+\frac{\beta }{2}t+\EE\left[\int_0^{t\wedge\theta_M}|\alpha^{(1)i}_s|^2ds\right]-\frac{\beta}{2}\EE\left[\int_0^{t\wedge\theta_M}(\tilde X_s^{(1)i})^2ds\right].
\en
By letting $t=T$ and $M\rightarrow \infty$, we have \eqref{nonexplosion} and 
\be\label{condition-X2}
\frac{\beta}{2}\EE\left[\int_0^{T}(\tilde X_s^{(1)i})^2ds\right]\leq (\tilde X_0^{(1)i})^2+\frac{\beta}{2}T+\EE\left[\int_0^T|\alpha^{(1)i}_s|^2ds\right].
\en
Applying Doob's martingale inequality and Cauchy-Schuwartz inequality to \eqref{coupled-appex}  and using \eqref{condition-X2} imply
\ba
\nonumber&&\EE[\sup_{t\leq s\leq T}|\tilde X^{(1)i}_s|^2]\\
\nonumber&\leq& 2\EE\left[\int_t^T|\gamma^{(1)}_s|ds\right]^2+2\EE\left[\int_t^T|\alpha^{(1)i}_s|ds\right]^2+2\EE\left[\sup_{t\leq u\leq T}\int_t^u2\sigma_1 \tilde X_s^{(1)i}dW^{(1)i}_s\right]^2\\
&\leq&C_1T\EE\left[\int_t^T|\gamma^{(1)}_s|^2+|\alpha^{(1)i}_s|^2ds\right]+C_2\EE\left[\int_t^T (\tilde X^{(1)i}_s)^2ds\right]< \infty,
\ea
where $C_1$ and $C_2$ are two positive constants. This proves \eqref{ui}. 
\end{proof}

\section{Proof of Theorem \ref{Hete-open}} \label{Appex-open}
Applying the Pontryagin principle to the proposed problem (\ref{objective}-\ref{diffusions}), we obtain the Hamiltonians written as 
\ba
\nonumber H^{(1)i}&=&\sum_{k=1}^{N_1}(\gamma_t^{(1)}+\alpha^{(1)k})y^{(1)i,(1)k}+\sum_{k=1}^{N_2}(\gamma^{(2)}_t+\alpha^{(2)k})y^{(1)i,(2)k}\\
&&+\frac{(\alpha^{(1)i})^2}{2}-q_1\alpha^{(1)i}(\overline x^{\lambda_1}-x^{(1)i})+\frac{\epsilon_1}{2}(\overline x^{\lambda_1}-x^{(1)i})^2,
\ea
and 
\ba
\nonumber H^{(2)j}&=&\sum_{k=1}^{N_1}(\gamma_t^{(1)}+\alpha^{(1)k})y^{(2)j,(1)k}+\sum_{k=1}^{N_2}(\gamma^{(2)}_t+\alpha^{(2)k})y^{(2)j,(2)k}\\
&&+\frac{(\alpha^{(2)j})^2}{2}-q_2\alpha^{(2)j}(\overline x^{\lambda_2}-x^{(2)j})+\frac{\epsilon_2}{2}(\overline x^{\lambda_2}-x^{(2)j})^2, 
\ea
where the adjoint diffusions $Y_t^{(1)i,(1)l}$, $Y_t^{(1)i,(2)h}$, $Y_t^{(2)j,(1)l}$, and $Y_t^{(2)j,(2)h}$ for $i,l=1,\cdots,N_1$ and $j,h=1,\cdots,N_2$  are given by
\ba
\label{Y-1-1}\nonumber dY_t^{(1)i,(1)l}&=&-\frac{\partial H^{(1)i}}{\partial x^{(1)l}}(\hat\alpha_t^{(1)i})dt+\sum_{k=0}^2Z_t^{(1)i,(1)l,k}dW_t^{(k)}\\
&&+\sum_{k=1}^{N_1}Z_t^{(1)i,(1)l,(1)k}dW_t^{(1)k}+\sum_{k=1}^{N_2}Z_t^{(1)i,(1)l,(2)k}dW_t^{(2)k},\\
\nonumber dY_t^{(1)i,(2)h}&=&-\frac{\partial H^{(1)i}}{\partial x^{(2)h}}(\hat\alpha_t^{(1)i})dt+\sum_{k=0}^2Z_t^{(1)i,(2)h,k}dW_t^{(k)}\\
&&+\sum_{k=1}^{N_1}Z_t^{(1)i,(2)h,(1)k}dW_t^{(1)k}+\sum_{k=1}^{N_2}Z_t^{(1)i,(2)h,(2)k}dW_t^{(2)k},
\ea
and
\ba
\nonumber dY_t^{(2)j,(1)l}&=&-\frac{\partial H^{(2)j}}{\partial x^{(1)l}}(\hat\alpha_t^{(2)j})dt+\sum_{k=0}^2Z_t^{(2)j,(1)l,k}dW_t^{(k)}\\
&&+\sum_{k=1}^{N_1}Z_t^{(2)j,(1)l,(1)k}dW_t^{(1)k}+\sum_{k=1}^{N_2}Z_t^{(2)j,(1)l,(2)k}dW_t^{(2)k},\\
\nonumber dY_t^{(2)j,(2)h}&=&-\frac{\partial H^{(2)j}}{\partial x^{(2)h}}(\hat\alpha_t^{(2)j})dt+\sum_{k=0}^2Z_t^{(2)j,(2)h,k}dW_t^{(k)}\\
\label{Y-1-4}&&+\sum_{k=1}^{N_1}Z_t^{(2)j,(2)h,(1)k}dW_t^{(1)k}+\sum_{k=1}^{N_2}Z_t^{(2)j,(2)h,(2)k}dW_t^{(2)k}
\ea
with the squared integrable progressive processes 
\ban
Z_t^{(1)i,(1)l,k}, Z_t^{(1)i,(1)l,(1)k_1}, Z_t^{(1)i,(1)l,(2)k_2}, Z_t^{(1)i,(2)h,k}, Z_t^{(1)i,(2)h,(1)k_1}, Z_t^{(1)i,(2)h,(2)k_2},\\
Z_t^{(2)j,(1)l,k}, Z_t^{(2)j,(1)l,(1)k}, Z_t^{(2)j,(1)l,(2)k}, Z_t^{(2)j,(2)h,k}, Z_t^{(2)j,(2)h,(1)k}, Z_t^{(2)j,(2)h,(2)k}, 
\ean
for $i,l=1,\cdots,N_1$, $j,h=1,\cdots,N_2$, and $k=0,1,2$. The terminal conditions are written as 
\[
Y_T^{(1)i,(1)l}=c_1\left( \frac{1-\lambda_1}{N_1}+\frac{\lambda_1}{N}-\delta_{(1)i,(1)l}\right)(\overline X^{\lambda_1}_T-X_T^{(1)i}),\,\;Y_T^{(1)i,(2)h}=c_1\frac{\lambda_1}{N}(\overline X^{\lambda_1}_T-X_T^{(1)i}),
\] 
and
\[
Y_T^{(2)j,(1)l}=c_2\frac{\lambda_2}{N}(\overline X^{\lambda_2}_T-X_T^{(2)j}),\;Y_T^{(2)j,(2)h}=c_2\left( \frac{1-\lambda_2}{N_2}+\frac{\lambda_2}{N}-\delta_{(2)j,(2)h}\right)(\overline X^{\lambda_2}_T-X_T^{(2)j}).
\]

Through minimizing the Hamiltonians with respect to $\alpha$ written as
\[
\frac{\partial H^{(1)i}}{\partial \alpha^{(1)i}}(\hat\alpha^{(1)i})=0, \quad \frac{\partial H^{(2)j}}{\partial \alpha^{(2)j}}(\hat\alpha^{(2)j})=0,
\]
the Nash equilibria are given by 
\ba
\label{optimal_open-1-1}\hat\alpha^{o,(1)i}&=&q_1(\overline x^{\lambda_1}-x^{(1)i})-y^{(1)i, (1)i},\\
\label{optimal_open-1-2}\hat\alpha^{o,(2)j}&=&q_2(\overline x^{\lambda_2}-x^{(2)j})-y^{(2)j,(2)j},
\ea
such that the optimal forward equations for banks are given by
\ba
\nonumber dX^{(1)i}_t &=& \left(q_1(\overline X_t^{\lambda_1}-X_t^{(1)i})-Y_t^{(1)i, (1)i} +\gamma^{(1)}_{t}\right)dt\\
&&+\sigma_1\left(\rho dW^{(0)}_t+\sqrt{1-\rho^2}\left(\rho_{1}dW_t^{(1)}+\sqrt{1-\rho^2_{1}}dW^{(1)i}_t\right)\right) ,
\ea
and 
\ba 
\nonumber dX^{(2)j}_t &=& \left(q_2(\overline X_t^{\lambda_2}-X_t^{(2)j})-Y_t^{(2)j,(2)j}+\gamma^{(2)}_{t}\right)dt\\
&&+\sigma_2\left(\rho dW^{(0)}_t+\sqrt{1-\rho^2}\left(\rho_{1}dW_t^{(2)}+\sqrt{1-\rho^2_{2}}dW^{(2)j}_t\right)\right).
\ea

Inserting \eqref{optimal_open-1-1} and \eqref{optimal_open-1-2} into (\ref{Y-1-1}-\ref{Y-1-4}), the adjoint processes are rewritten as 
\ba
\nonumber dY_t^{(1)i,(1)l}&=&\left( \frac{1-\lambda_1}{N_1}+\frac{\lambda_1}{N}-\delta_{(1)i,(1)l}\right)\left\{-(\epsilon_1-q_1^2)(\overline X^{\lambda_1}_t-X_t^{(1)i})-q_1Y_t^{(1)i,(1)i}\right\}dt\\
\nonumber&&+\sum_{k=0}^2Z_t^{(1)i,(1)l,k}dW_t^{(k)}\\
&&+\sum_{k_1=1}^{N_1}Z_t^{(1)i,(1)l,(1)k_1}dW_t^{(1)k_1}+\sum_{k_2=1}^{N_2}Z_t^{(1)i,(1)l,(2)k_2}dW_t^{(2)k_2},
\label{Y-1-1}\\
\nonumber dY_t^{(1)i,(2)h}&=&\frac{\lambda_1}{N}\left\{-(\epsilon_1-q_1^2)(\overline X^{\lambda_1}_t-X_t^{(1)i})-q_1Y_t^{(1)i,(1)i}\right\}dt+\sum_{k=0}^2Z_t^{(1)i,(2)h,k}dW_t^{(k)}\\
 &&+\sum_{k_1=1}^{N_1}Z_t^{(1)i,(2)h,(1)k_1}dW_t^{(1)k_1}+\sum_{k_2=1}^{N_2}Z_t^{(1)i,(2)h,(2)k_2}dW_t^{(2)k_2},
\ea
with the terminal conditions 
\[
Y_T^{(1)i,(1)l}=c_1\left( \frac{1-\lambda_1}{N_1}+\frac{\lambda_1}{N}-\delta_{(1)i,(1)l}\right)(\overline X^{\lambda_1}_T-X_T^{(1)i}),\,\;Y_T^{(1)i,(2)h}=c_1\frac{\lambda_1}{N}(\overline X^{\lambda_1}_T-X_T^{(1)i}),
\] 
and
\ba
\nonumber dY_t^{(2)j,(1)l}&=&\frac{\lambda_2}{N}\left\{-(\epsilon_2-q_2^2)(\overline X^{\lambda_2}_t-X_t^{(2)j})-q_1Y_t^{(1)i,(1)i}\right\}dt+\sum_{k=0}^2Z_t^{(2)j,(1)l,k}dW_t^{(k)}\\
&&+\sum_{k_1=1}^{N_1}Z_t^{(2)j,(1)l,(1)k_1}dW_t^{(1)k_1}+\sum_{k_2=1}^{N_2}Z_t^{(2)j,(1)l,(2)k_2}dW_t^{(2)k_2},\\
\nonumber dY_t^{(2)j,(2)h}&=&\left( \frac{1-\lambda_2}{N_2}+\frac{\lambda_2}{N}-\delta_{(2)j,(2)h}\right)\left\{-(\epsilon_2-q_2^2)(\overline X^{\lambda_2}_t-X^{(2)j})-q_2Y_t^{(2)j,(2)j}\right\}dt\\
\nonumber&&+\sum_{k=0}^2Z_t^{(2)j,(2)h,k}dW_t^{(k)}\\
&&+\sum_{k_1=1}^{N_1}Z_t^{(2)j,(2)h,(1)k_1}dW_t^{(1)k_1}+\sum_{k_2=1}^{N_2}Z_t^{(2)j,(2)h,(2)k_2}dW_t^{(2)k_2},
\label{Y-1-4}
\ea 
with the terminal conditions 
\[
Y_T^{(2)j,(1)l}=c_2\frac{\lambda_2}{N}(\overline X^{\lambda_2}_T-X_T^{(2)j}),\;Y_T^{(2)j,(2)h}=c_2\left( \frac{1-\lambda_2}{N_2}+\frac{\lambda_2}{N}-\delta_{(2)j,(2)h}\right)(\overline X^{\lambda_2}_T-X_T^{(2)j}).
\]
We then make the ansatz written as
\ba
\nonumber Y_t^{(1)i,(1)l}&=&\left(\frac{1}{\widetilde N_1}-\delta_{(1)i,(1)l}\right)\left(\eta_t^{(o),1}(\overline X_t^{(1)}-X_t^{(1)i})+\eta_t^{(o),2}\overline X_t^{(1)}+\eta_t^{(o),3}\overline X_t^{(2)}+\eta_t^{(o),4}\right)\\
\label{ansatz_open_1} \\
Y_t^{(1)i,(2)h}&=&\frac{\lambda_1}{N}\left(\eta_t^{(o),1}(\overline X_t^{(1)}-X_t^{(1)i})+\eta_t^{(o),2}\overline X_t^{(1)}+\eta_t^{(o),3}\overline X_t^{(2)}+\eta_t^{(o),4}\right)
\ea
and
\ba
 Y_t^{(2)j,(1)l}&=&\frac{\lambda_2}{N}\left(\phi_t^{o,(1)}(\overline X_t^{(1)}-X_t^{(1)i})+\phi_t^{o,(2)}\overline X_t^{(1)}+\phi_t^{o,(3)}\overline X_t^{(2)}+\phi_t^{o,(4)}\right)\\
\nonumber  Y_t^{(2)j,(2)h}&=&\left(\frac{1}{\widetilde N_2}-\delta_{(1)i,(1)l}\right)\left(\phi_t^{o,(1)}(\overline X_t^{(1)}-X_t^{(1)i})+\phi_t^{o,(2)}\overline X_t^{(1)}+\phi_t^{o,(3)}\overline X_t^{(2)}+\phi_t^{o,(4)}\right)\\
  \label{ansatz_open_2} 
\ea
where $$\frac{1}{\widetilde N_k}=\frac{1-\lambda_k}{N_k}+\frac{\lambda_k}{N},$$ for $k=1,2$. Differentiating (\ref{ansatz_open_1}-\ref{ansatz_open_2}) and identifying the $dY_t^{(1)i,(1)l}$, $dY_t^{(1)i,(2)h}$, $dY_t^{(2)j,(1)l}$, and $Y_t^{(2)j,(2)h}$ with (\ref{Y-1-1}-\ref{Y-1-4}), we obtain that the deterministic functions $\eta_t^{o,(i)}$ and $\phi_t^{o,(i)}$ for $i=1,\cdots,4$  must satisfy
\ba
\label{eta_open-1}\dot\eta_t^{o,(1)}&=&\left(2-\frac{1}{\widetilde N_1}\right)q_1\eta_t^{o,(1)}+\left(1-\frac{1}{\widetilde N_1}\right)(\eta_t^{o,(1)})^2-(\epsilon_1-q_1^2)\\
\nonumber \dot\eta_t^{o,(2)}&=&-\left(q_1\lambda_1(\beta_1-1)+(1-\frac{1}{\widetilde N_1})\eta_t^{o,(2)}\right)\eta_t^{o,(2)}-\left(q_2\lambda_1\beta_1+(1-\frac{1}{\widetilde N_2})\phi^{o,(2)}\right)\eta_t^{o,(3)}\\
&&-q_1\left(\frac{1}{\widetilde N_1}-1\right)\eta_t^{o,(2)}-(\epsilon_1-q_1^2)\lambda_1(\beta_1-1)\\
\nonumber \dot\eta_t^{o,(3)}&=&-\left(q_1\lambda_2\beta_2+(1-\frac{1}{\widetilde N_1})\eta_t^{o,(3)}\right)\eta_t^{o,(2)}-\left(q_2\lambda_2(\beta_1-1)+(1-\frac{1}{\widetilde N_2})\phi_t^{o,(3)}\right)\eta_t^{o,(3)}\\
&&-q_1\left(\frac{1}{\widetilde N_1}-1\right)\eta_t^{o,(3)}-(\epsilon_1-q_1^2)\lambda_1\beta_2\\
\nonumber \dot\eta_t^{o,(4)}&=&-\left((1-\frac{1}{\widetilde N_1})\eta_t^{o,(4)}+\gamma_t^{(1)}\right)\eta_t^{o,(2)}\\
 &&-\left((1-\frac{1}{\widetilde N_2})\phi_t^{o,(4)}+\gamma_t^{(2)}\right)\eta_t^{o,(3)}-q_1\left(\frac{1}{\widetilde N_1}-1\right)\eta_t^{o,(4)}
\ea
\ba
\dot\phi_t^{o,(1)}&=&\left(2-\frac{1}{\widetilde N_2}\right)q_2\phi_t^{o,(1)}+\left(1-\frac{1}{\widetilde N_2}\right)(\phi_t^{o,(1)})^2-(\epsilon_2-q_2^2)\\
\nonumber \dot\phi_t^{o,(2)}&=&-\left(q_2\lambda_2(\beta_1-1)+(1-\frac{1}{\widetilde N_1})\eta_t^{o,(2)}\right)\phi_t^{o,(2)}-\left(q_2\lambda_2\beta_2+(1-\frac{1}{\widetilde N_2})\phi_t^{o,(2)}\right)\phi_t^{o,(3)}\\
&&-q_2\left(\frac{1}{\widetilde N_2}-1\right)\phi_t^{o,(2)}-(\epsilon_2-q_2^2)\lambda_2\beta_1\\
\nonumber\dot\phi_t^{o,(3)}&=&-\left(q_1\lambda_2\beta_2+(1-\frac{1}{\widetilde N_1})\eta_t^{o,(3)}\right)\phi_t^{o,(2)}-\left(q_2\lambda_2(\beta_1-1)+(1-\frac{1}{\widetilde N_2})\phi_t^{o,(3)}\right)\phi_t^{o,(3)}\\
&&-q_2\left(\frac{1}{\widetilde N_2}-1\right)\phi_t^{o,(3)}-(\epsilon_2-q_2^2)\lambda_1(\beta_2-1)\\
\nonumber\dot\phi_t^{o,(4)}&=&-\left((1-\frac{1}{\widetilde N_1})\eta_t^{o,(4)}+\gamma_t^{(1)}\right)\phi_t^{o,(2)}\\
&&-\left((1-\frac{1}{\widetilde N_2})\phi_t^{o,(4)}+\gamma_t^{(2)}\right)\phi_t^{o,(3)}-q_2\left(\frac{1}{\widetilde N_2}-1\right)\phi_t^{o,(4)}\label{phi_open-4}
\ea
with the terminal conditions 
\ban
\eta_T^{o,(1)}=c_1,\;\eta_T^{o,(2)}=c_1\lambda_1(\beta_1-1)\;\eta_T^{o,(3)}=c_1\lambda_1\beta_2,\;\eta_T^{o,(4)}=0,\\
\phi_T^{o,(1)}=c_2,\;\eta_T^{o,(2)}=c_2\lambda_2\beta_1\;\phi_T^{o,(3)}=c_2\lambda_2(\beta_1-1),\;\phi_T^{o,(4)}=0,
\ean
and the squared integrable progressive processes are given by
\ban
&&Z_t^{(1)i,(1)l,0}=\left(\frac{1}{\widetilde N_1}-\delta_{(1)i,(1)l}\right)\rho\left(\sigma_1\eta_t^{o,(2)}+\sigma_2\eta_t^{o,(3)}\right),\\
&&Z_t^{(1)i,(1)l,1}=\left(\frac{1}{\widetilde N_1}-\delta_{(1)i,(1)l}\right)\eta_t^{o,(2)}\sigma_1\sqrt{1-\rho^2}\rho_1,\;\\
&&Z_t^{(1)i,(1)l,2}=\left(\frac{1}{\widetilde N_1}-\delta_{(1)i,(1)l}\right)\eta_t^{o,(3)}\sigma_2\sqrt{1-\rho^2}\rho_2,\\
&&Z_t^{(1)i,(1)l,(1)k_1}=\frac{1}{N_1}\left(\frac{1}{\widetilde N_1}-\delta_{(1)i,(1)l}\right)\eta_t^{o,(1)}\left(\frac{1}{N_1}\sigma_1\sqrt{1-\rho^2}\sqrt{1-\rho_1^2}-\delta_{(1)i,(1)k_1}\right),\\
&&Z_t^{(1)i,(1)l,(2)k_2}=\frac{1}{N_1}\left(\frac{1}{\widetilde N_1}-\delta_{(1)i,(1)l}\right)\eta_t^{o,(3)}\sigma_2\sqrt{1-\rho^2}\sqrt{1-\rho_2^2},
\ean
and
\ban
&&Z_t^{(1)i,(2)h,0}=\frac{\lambda_1}{N}\rho\left(\sigma_1\eta_t^{o,(2)}+\sigma_2\eta_t^{o,(3)}\right),\\
&&Z_t^{(1)i,(2)h,1}=\frac{\lambda_1}{N}\eta_t^{o,(2)}\sigma_1\sqrt{1-\rho^2}\rho_1,\\
&&Z_t^{(1)i,(2)h,2}=\frac{\lambda_1}{N}\eta_t^{o,(3)}\sigma_2\sqrt{1-\rho^2}\rho_2,\\
&&Z_t^{(1)i,(2)h,(1)k_1}=\frac{\lambda_1}{N}\eta_t^{o,(1)}\left(\frac{1}{N_1}\sigma_1\sqrt{1-\rho^2}\sqrt{1-\rho_1^2}-\delta_{(1)i,(1)k_1}\right)\\
&&Z_t^{(1)i,(2)h,(2)k_2}=\frac{\lambda_1}{N}\eta_t^{o,(3)}\sigma_2\sqrt{1-\rho^2}\sqrt{1-\rho_2^2},\\
\ean
and
\ban
&&Z_t^{(2)j,(1)l,0}=\frac{\lambda_2}{N}\rho\left(\sigma_1\phi_t^{o,(2)}+\sigma_2\phi_t^{o,(3)}\right),\\
&&Z_t^{(2)j,(1)l,1}=\frac{\lambda_2}{N}\phi_t^{o,(2)}\sigma_1\sqrt{1-\rho^2}\rho_1,\\
&&Z_t^{(2)j,(1)l,2}=\frac{\lambda_2}{N}\phi_t^{o,(3)}\sigma_2\sqrt{1-\rho^2}\rho_2,\\
&&Z_t^{(2)j,(1)l,(1)k_1}=\frac{\lambda_2}{N}\phi_t^{o,(2)} \frac{1}{N_1}\sigma_1\sqrt{1-\rho^2}\sqrt{1-\rho_1^2},\\
&&Z_t^{(2)j,(1)l,(2)k_2}=\frac{\lambda_2}{N}\phi_t^{o,(1)}\left( \frac{1}{N_1}\sigma_1\sqrt{1-\rho^2}\sqrt{1-\rho_2^2}-\delta_{(2)j,(2)k_2}\right),
\ean
and
\ban
&&Z_t^{(2)j,(2)h,0}=\left(\frac{1}{\widetilde N_1}-\delta_{(2)j,(2)h}\right)\rho\left(\sigma_1\phi_t^{o,(2)}+\sigma_2\phi_t^{o,(3)}\right),\\
&&Z_t^{(1)i,(2)h,1}=\left(\frac{1}{\widetilde N_1}-\delta_{(2)j,(2)h}\right)\phi_t^{o,(2)}\sigma_1\sqrt{1-\rho^2}\rho_1,\\
&&Z_t^{(2)j,(2)h,2}=\left(\frac{1}{\widetilde N_1}-\delta_{(2)j,(2)h}\right)\frac{\lambda_2}{N}\phi_t^{o,(3)}\sigma_2\sqrt{1-\rho^2}\rho_2,\\
&&Z_t^{(2)j,(2)h,(1)k_1}=\left(\frac{1}{\widetilde N_1}-\delta_{(2)j,(2)h}\right)\phi_t^{o,(2)} \frac{1}{N_1}\sigma_1\sqrt{1-\rho^2}\sqrt{1-\rho_1^2},\\
&&Z_t^{(2)j,(2)h,(2)k_2}=\frac{\lambda_2}{N}\phi_t^{o,(1)}\left( \frac{1}{N_1}\sigma_1\sqrt{1-\rho^2}\sqrt{1-\rho_2^2}-\delta_{(2)j,(2)k_2}\right),
\ean
for $i,l=1,\cdots,N_1$ and $j,h=1,\cdots,N_2$. Hence, the open-loop Nash equilibria are written as 
\ba
\nonumber \hat\alpha^{o,(1)i}&=&\left(q_1+(1-\frac{1}{\widetilde N_1})\eta_t^{o,(1)}\right)(\overline X_t^{(1)}-X_t^{(1)i})+\left(q_1\lambda_1(\beta_1-1)+(1-\frac{1}{\widetilde N_1})\eta_t^{o,(2)}\right)\overline X_t^{(1)}\\
\label{open-app-1}&&+\left(q_1\lambda_1\beta_2+(1-\frac{1}{\widetilde N_1})\eta_t^{o,(3)}\right)\overline X_t^{(2)}+\left(1-\frac{1}{\widetilde N_1}\right)\eta_t^{o,(4)},\\
\nonumber \hat\alpha^{o,(2)j}&=&\left(q_2+(1-\frac{1}{\widetilde N_2})\phi_t^{o,(1)}\right)(\overline X_t^{(2)}-X_t^{(2)j})+\left(q_2\lambda_2\beta_1+(1-\frac{1}{\widetilde N_2})\phi_t^{o,(2)}\right)\overline X_t^{(1)}\\
\label{open-app-2}&&+\left(q_2\lambda_2(\beta_2-1)+(1-\frac{1}{\widetilde N_2})\phi_t^{o,(3)}\right)\overline X_t^{(2)}+\left(1-\frac{1}{\widetilde N_2}\right)\phi_t^{o,(4)}.
\ea
Note that the existence of the coupled ODEs (\ref{eta_open-1}-\ref{phi_open-4}) in the case of sufficiently large $N_1$ and $N_2$ is studied in Proposition \ref{Prop_suff}. Based on the open-loop equilibria (\ref{open-app-1}-\ref{open-app-2}), we have the lending and borrowing system satisfying the Lipchitz condition in the sense that the existence of the corresponding FBSDEs can be verified using the fixed point argument. See \cite{Carmona-Fouque2016} for instance.

\section{Proof of Theorem \ref{Hete-MFG-prop}}\label{Appex-Hete-MFG}
Due to the non-Markovian structure for the given $m_t^{(k)}$ for $k=1,\cdots,d$, in order to obtain the $\epsilon$-Nash equilibrium for the coupled diffusions with common noises, we again apply the adjoint FBSDEs discussed in  \cite{CarmonaDelarueLachapelle} and \cite{R.Carmona2013}. The corresponding Hamiltonian is given by
\be
H^{k}(t,x,y^{(k)},\alpha)=\sum_{h=1}^d(\alpha^{(h)}+\gamma^{(h)}_t)y^{(k),h}+\frac{(\alpha^{(k)})^2}{2}-q_k\alpha^{(k)}\left(M^{\lambda_k}_t-x^{(k)}\right)+\frac{\epsilon_k}{2}\left(M^{\lambda_k}_t-x^{(k)}\right)^2,
\en   
where $x=(x^{(1)},\cdots,x^{(d)})$, $y^{(k)}=(y^{k,1},\cdots,y^{k,d})$, and $\alpha=(\alpha^{(1)},\cdots,\alpha^{(d)})$. The Hamiltonian attains its minimum at
\be
\hat\alpha^{m,(k)}_t=q_k\left(M^{\lambda_k}_t-x^{(k)}\right)-y^{k,k}.
\en 
The backward equations satisfy 
\ba
\nonumber dY^{k,l}_t&=&-\partial_{x^{(l)}}H^k(\hat\alpha^{(k)})dt+\sum_{h=0}^dZ^{0,k,l,h}_tdW^{(0),(h)}_t+\sum_{h=1}^dZ^{k,l,h}_tdW^{(h)}_t\\
\nonumber &=&(q_kY^{k,k}_t+(\epsilon_k-q_k^2)(M^{\lambda_k}_t-X^{(k)}_t))\delta_{k,l}dt+\sum_{h=0}^dZ^{0,k,l,h}_tdW^{(0),(h)}_t+\sum_{h=1}^dZ^{k,l,h}_tdW^{(h)}_t,\\
 \label{Y-MFG}
\ea
with the terminal conditions $Y^{k,l}_T=\frac{c_k}{2}(X^{(k)}_T-m^{(k)}_T)\delta_{k,l}$ for $k,l=1,\cdots,d$ where the processes $Z^{0,k,l,h}_t$ and $Z^{k,l,h}_t$ are adapted and square integrable.  We make the ansatz for $Y^{k,l}_t$ written as
\be\label{ansatz-MFG}
Y^{k,l}_t=-\left(\eta^{m,(k)}_t(m^{(k)}_t-X_t^{(k)})+\sum_{h_1=1}^d\psi_t^{m,(k),h_1}m^{(h_1)}_t+\mu^{m,(k)}_t\right)\delta_{k,l},
\en
leading to 
\ba
\nonumber dX^{(k)}_t&=&\bigg\{(q_k+\eta^{m,(k)}_t)(m^{(k)}_t -X_t^{(k)})+\sum_{h_1=1}^d\psi_t^{m,(k),h_1}m^{(h_1)}_t+\mu^{m,(k)}_t+\gamma^{(k)}_t\\
\nonumber&&+q_k\lambda_k\sum_{h_1=1}^d(\beta_{h_1}-\delta_{k,h_1})m^{(h_1)}_t\bigg\}dt\\
&&+\sigma_k\left(\rho dW^{(0),(0)}_t+\sqrt{1-\rho^2}\left(\rho_{k}dW_t^{(0),(k)}+\sqrt{1-\rho^2_{k}}dW^{(k)}_t\right)\right),\label{X-MFG-1}\\
\nonumber dm^{(k)}_t&=&\bigg\{\sum_{h_1=1}^d\psi_t^{m,(k),h_1}m^{(h_1)}_t+\mu^{m,(k)}_t+\gamma^{(k)}_t+q_k\lambda_k\sum_{h_1=1}^d(\beta_{h_1}-\delta_{k,h_1})m^{(h_1)}_t\bigg\}dt\\
&&+\sigma_k\left(\rho dW^{(0),(0)}_t+\sqrt{1-\rho^2} \rho_{k}dW_t^{(0),(k)} \right) \label{m-hete}
\ea
Inserting the ansatz \eqref{ansatz-MFG} into \eqref{Y-MFG} gives 
\ba
\nonumber dY^{k,l}_t&=&\delta_{k,l}\bigg\{(-q_k\eta_t^{m,(k)}+\epsilon_k-q_k^2)(m^{(k)}_t-X_t^{(k)})+(\epsilon_k-q_k^2) \lambda_k\sum_{h_1=1}^d(\beta_{h_1}-\delta_{k,h_1})m_t^{(h_1)} \\
\nonumber&&-q_k\sum_{h_1=1}^d\psi_t^{m,(k),h_1}m^{(h_1)}_t-q_k\mu^{m,(k)}_t\bigg\}dt+\sum_{h=0}^dZ^{0,k,l,h}_tdW^{(0),(h)}_t+\sum_{h=1}^dZ^{k,l,h}_tdW^{(h)}_t,\\\label{Y-MFG-1}
\ea
and applying It\^o formula to \eqref{ansatz-MFG} and using \eqref{X-MFG-1} and \eqref{m-hete} imply
\ba
\nonumber dY^{k,l}_t&=&\delta_{k,l}\bigg\{\bigg(-\dot\eta^{m,(k)}_t(m^{(k)}_t-X_t^{(k)})+\eta^{m,(k)}_t (q_k+\eta_t^{m,(k)})(m^{(k)}_t-X^{(k)}_t)\\
\nonumber&&-\dot\mu_t^{m,(k)}-\sum_{h_1=1}^d\dot\psi_t^{m,(k),h_1}m^{(h_1)}_t\\
 \nonumber&& -\sum_{h=1}^d\psi_t^{m,(k),h}\left(\sum_{h_1=1}^d(\psi_t^{m,(h),h_1}+q_h\lambda_h(\beta_{h_1}-\delta_{h,h_1}))m^{(h_1)}_t+\mu_t^{m,(h)}+\gamma^{(h)}_t \right)\bigg)dt \\
\nonumber&&+\sum_{h=1}^d\psi_t^{m,(k),h}\sigma_h\left(\rho dW^{(0),(0)}_t+\sqrt{1-\rho^2} \rho_{h}dW_t^{(0),(h)}\right)\\
&&+\eta_t^{m,(k)}\sigma_k\sqrt{1-\rho^2}\sqrt{1-\rho_k^2}dW^{(k)}_t \bigg\}.
\label{Y-MFG-2}
\ea
Similarly, through identifying \eqref{Y-MFG-1} and \eqref{Y-MFG-2}, we get $\eta_t^{(k)}$, $\psi^{(k),h}_t$, and $\mu_t^{(k)}$ must satisfy (\ref{Hete-eta-MFG}-\ref{Hete-mu-MFG}) and the squared integrable processes  $Z^{0,k,l,h}_t$ and $Z^{k,l,h}_t$ satisfying 
 \be
 Z^{0,k,l,0}_t=-\eta_t^{m,(k)}\lambda_k\rho\sum_{h_1=1}^d\sigma_{h_1}(\beta_{h_1}-\delta_{k,h_1}+\psi_t^{m,(k),h_1}),\;l=k,\quad Z^{0,k,l,0}_t=0, \; l\neq k,
 \en
 and 
 \be
Z^{0,k,l,h}_t=-\eta_t^{m,(k)}\lambda_k\sqrt{1-\rho^2}\sigma_{h }(\beta_{h }-\delta_{k,h}+\psi_t^{m,(k),h}),\;l=k,\quad Z^{0,k,l,h}_t=0, \; l\neq k,
 \en
 and
 \be
Z^{k,l,h}_t=\eta^{m,(k)}_t\sigma_k\sqrt{1-\rho^2}\sqrt{1-\rho_k^2},\;l=k,\quad Z^{k,l,h}_t=0,\; l\neq k.
 \en
By the fixed point argument, the $\epsilon$-Nash equilibria are given by
\be
\hat\alpha_t^{m,(k)}=(q_k+\eta^{m,(k)}_t)(m^{(k)}_t-x^{(k)})+\sum_{h=1}^d\widetilde\psi_t^{m,(k),h}m^{h}_t+\mu^{m,(k)}_t,\quad k=1,\cdots,d
\en
where $\widetilde\psi_t^{m,(k),h}=\psi_t^{m,(k),h}+q_k\lambda_k(\beta_k-\delta_{k,h})$. 

We now study the existence of the coupled ODEs (\ref{Hete-eta-MFG}-\ref{Hete-mu-MFG}). Note that we show the existence of the case of two heterogeneous groups in Proposition \ref{Prop_suff}. Observe that \eqref{Hete-eta-MFG} satisfies the Riccati equation without coupling. Given \eqref{Hete-psi-MFG}, the system \eqref{Hete-mu-MFG} is the system of linear ODEs. Hence, it is sufficient to show the existence of \eqref{Hete-psi-MFG}. Similar to the results in Proposition \ref{Prop_suff},  in the general $d$ groups, we obtain 
\[
\sum_{h_1=1}^d\psi_t^{m,(k),h_1}=0,
\] 
implying that given $k=h_1$
\be\label{MFG_suff_cond_1}
\psi_t^{m,(k),k}=-\sum_{h_1\neq k}\psi_t^{m,(k),h_1}. 
\en
Now, by inserting \eqref{MFG_suff_cond_1} into \eqref{Hete-mu-MFG}, for $k\neq h_1$, \eqref{Hete-mu-MFG} can be rewritten as 
\ba
\nonumber   \dot\psi_t^{m,(k),h_1}&=&q_k\psi_t^{m,(k),h_1}+\sum_{h\neq k}\psi_t^{m,(k),h}\left(\psi_t^{m,(k),h_1}+q_k\lambda_k\beta_{h_1 }\right)\\
\nonumber &&-\sum_{h\neq k}\psi_t^{m,(k),h}\left(\psi_t^{m,(h),h_1}+q_h\lambda_h(\beta_{h_1 }-\delta_{h,h_1})\right)-(\epsilon_k-q_k^2)\lambda_k\beta_{h_1} \\
\nonumber  &=&q_k(1+\lambda_k\beta_{h_1})\psi_t^{m,(k),h_1}+(\psi_t^{m,(k),h_1})^2\\
\nonumber &&-\sum_{h\neq k,h_1}\psi_t^{m,(k),h}\left(\psi_t^{m,(h),h_1}+q_h\lambda_h\beta_{h_1 }\right)\\
\nonumber  &&+\psi_t^{m,(k),h_1}\left(\sum_{h\neq h_1}\psi_t^{m,(k),h}+q_{h_1}\lambda_{h_1}(1-\beta_{h_1 })\right)\\
\nonumber  &&+\sum_{h\neq k, h_1}\psi_t^{m,(k),h}\left(\psi_t^{m,(k),h_1}+q_k\lambda_k\beta_{h_1 }\right)-(\epsilon_k-q_k^2)\lambda_k\beta_{h_1}\\
\nonumber  &=&q_k(1+\lambda_k\beta_{h_1})\psi_t^{m,(k),h_1}+(\psi_t^{m,(k),h_1})^2\\
\nonumber  &&+\psi_t^{m,(k),h_1}\left(\sum_{h\neq h_1}\psi_t^{m,(k),h}+q_{h_1}\lambda_{h_1}(1-\beta_{h_1 })\right)+\psi_t^{m,(k),h_1}\sum_{h\neq k, h_1}\psi_t^{m,(k),h}\\
  &&-\sum_{h\neq k,h_1}\psi_t^{m,(k),h}\left(\psi_t^{m,(h),h_1}+q_h\lambda_h\beta_{h_1 }-q_k\lambda_k\beta_{h_1 }\right)-(\epsilon_k-q_k^2)\lambda_k\beta_{h_1}\
\label{C-15}
\ea
We now further assume  $c_{\tilde k}\geq \max_{k,h}\left(\frac{q_k\lambda_k}{\lambda_h}-q_h\right)$ for $\tilde k=1,\cdots,d$ such that the term 
\be\label{C-15-1}
\sum_{h\neq k,h_1}\psi_t^{m,(k),h}\left(\psi_t^{m,(h),h_1}+q_h\lambda_h\beta_{h_1 }-q_k\lambda_k\beta_{h_1 }\right)
\en
 stays in negative for all $t$ in order to guarantee $\psi_t^{m,(k),h_1}\geq 0$ for $k\neq h_1$. Note that in the two-group case with $d=2$, the equation \eqref{C-15-1} can be removed. See Proposition \ref{Prop_suff} for details.   

Similarly, using $\check\psi_t^{m,(k),h_1}= \psi_{T-t}^{m,(k),h_1}$, we have 
 \ban
  \dot{\check\psi}_t^{m,(k),h_1}&=&-q_k(1+\lambda_k\beta_{h_1})\check\psi_t^{m,(k),h_1}-(\check\psi_t^{m,(k),h_1})^2-\sum_{h\neq k, h_1}\check\psi_t^{m,(k),h}\left(\check\psi_t^{m,(k),h_1}+q_k\lambda_k\beta_{h_1 }\right)\\
   &&-\check\psi_t^{m,(k),h_1}\left(\sum_{h\neq h_1}\check\psi_t^{m,(k),h}+q_{h_1}\lambda_{h_1}(1-\beta_{h_1 })\right)\\
  &&+\sum_{h\neq k,h_1}\check\psi_t^{m,(k),h}\left(\check\psi_t^{m,(h),h_1}+q_h\lambda_h\beta_{h_1 } \right)+(\epsilon_k-q_k^2)\lambda_k\beta_{h_1} \\
  &\leq &-q_k(1+\lambda_k\beta_{h_1})\check\psi_t^{m,(k),h_1}-(\check\psi_t^{m,(k),h_1})^2\\
  &&+\sum_{h\neq k,h_1}\check\psi_t^{m,(k),h}\left(\check\psi_t^{m,(h),h_1}+q_h\lambda_h\beta_{h_1 }\right)+(\epsilon_k-q_k^2)\lambda_k\beta_{h_1} 
 \ean
Let $$\underline \zeta=\min_{k,h}q_k(1+\lambda_k\beta_h),\quad\overline \zeta=\max_{k,h}q_k\lambda_k\beta_h.$$ We have 
\ba\label{psi_ij}
\nonumber \dot{\check\psi}_t^{m,(k),h_1}&\leq&\underline\zeta \check\psi_t^{m,(k),h_1}-(\check\psi_t^{m,(k),h_1})^2+\frac{1}{2}\sum_{h\neq k,h_1}\left((\check\psi_t^{m,(k),h})^2+(\check\psi_t^{m,(h),h_1})^2\right)\\
 &&+\overline \zeta\sum_{h\neq k,h_1}\check\psi_t^{m,(k),h}+(\epsilon_k-q_k^2)\lambda_k\beta_{h_1} .
\ea
Now, using $$\check\psi_t=\sum_{k=1,\cdots,d,h_1=1,\cdots,N_k,k\neq h_1}\check\psi_t^{m,(k),h_1}$$ and \eqref{psi_ij} leads to 
\ba
\dot{\check\psi}_t&\leq&-\underline\zeta\check\psi_t+\overline \zeta\check\psi_t=(\overline \zeta-\underline \zeta)\check\psi_t+\hat\zeta
\ea 
implying 
\be
\check\psi_t\leq\widetilde\zeta e^{(\overline \zeta-\underline \zeta)t}+\frac{\hat\zeta}{\overline \zeta-\underline \zeta}\left(e^{(\overline \zeta-\underline \zeta)t}-1\right),
\en
where $$\widetilde\zeta=\sum_{k=1,\cdots,d,h_1=1,\cdots,N_k,k\neq h_1}c_k\lambda_k\beta_{h_1},\quad \hat\zeta=\sum_{k=1,\cdots,d,h_1=1,\cdots,N_k,k\neq h_1}(\epsilon_k-q_k^2)\lambda_k\beta_{h_1}.$$
Using $\psi_t^{m,(k),h_1}\geq 0$ for $k\neq h_1$, the proof is complete. \qed

%


\bibliographystyle{plainnat}
\bibliography{references-delay-games-feller}

\begin{thebibliography}{27}
\providecommand{\natexlab}[1]{#1}
\providecommand{\url}[1]{\texttt{#1}}
\expandafter\ifx\csname urlstyle\endcsname\relax
  \providecommand{\doi}[1]{doi: #1}\else
  \providecommand{\doi}{doi: \begingroup \urlstyle{rm}\Url}\fi

\bibitem[Bensoussan et~al.(2016)Bensoussan, Sung, Yam, and
  Yung]{Bensoussan_et_al}
A.~Bensoussan, K.C.J. Sung, S.C.P. Yam, and S.P. Yung.
\newblock Linear quadratic mean field games.
\newblock \emph{Journal of Optimization Theory and Applications}, 169\penalty0
  (2):\penalty0 496--529, 2016.

\bibitem[Biagini et~al.(2019)Biagini, Mazzon, and Meyer-Brandis]{BMMB2019}
F.~Biagini, A.~Mazzon, and T.~Meyer-Brandis.
\newblock Financial asset bubbles in banking networks.
\newblock \emph{SIAM Journal on Financial Mathematics}, 10\penalty0
  (2):\penalty0 430--465, 2019.

\bibitem[Carmona(2016)]{CarmonaSIAM2016}
R.~Carmona.
\newblock \emph{Lectures on BSDEs, Stochastic Control, and Stochastic
  Differential Games with Financial Applications}.
\newblock SIAM Book Series in Financial Mathematics 1, 2016.

\bibitem[Carmona and Delaure(2018{\natexlab{a}})]{MFG-book-1}
R.~Carmona and Francois Delaure.
\newblock \emph{Probabilistic Theory of Mean Field Games with Applications I}.
\newblock Springer International Publishing, 2018{\natexlab{a}}.

\bibitem[Carmona and Delaure(2018{\natexlab{b}})]{MFG-book-2}
R.~Carmona and Francois Delaure.
\newblock \emph{Probabilistic Theory of Mean Field Games with Applications II}.
\newblock Springer International Publishing, 2018{\natexlab{b}}.

\bibitem[Carmona and Lacker(2015)]{Carmona-Lacker2015}
R.~Carmona and D.~Lacker.
\newblock A probabilistic weak formulation of mean field games and
  applications.
\newblock \emph{The Annals of Applied Probability}, 25\penalty0 (3):\penalty0
  1189--1231, 2015.

\bibitem[Carmona et~al.(2013)Carmona, Delarue, and
  Lachapelle]{CarmonaDelarueLachapelle}
R.~Carmona, F.~Delarue, and A.~Lachapelle.
\newblock Control of {M}c{K}ean-{V}lasov versus {M}ean {F}ield {G}ames.
\newblock \emph{Mathematics and Financial Economics}, 7:\penalty0 131--166,
  2013.

\bibitem[Carmona et~al.(2015)Carmona, Fouque, and Sun]{R.Carmona2013}
R.~Carmona, J.-P. Fouque, and L.-H. Sun.
\newblock Mean field games and systemic risk.
\newblock \emph{Communications in Mathematical Sciences}, 13\penalty0
  (4):\penalty0 911--933, 2015.

\bibitem[Carmona et~al.(2018)Carmona, Fouque, Mousafa, and
  Sun]{Carmona-Fouque2016}
R.~Carmona, J.-P. Fouque, M.~Mousafa, and L.-H. Sun.
\newblock Systemic risk and stochastic games with delay.
\newblock \emph{Journal of Optimization Theory and Applications}, 179\penalty0
  (2):\penalty0 366--399, 2018.

\bibitem[Espinosa and Touzi(2015)]{Touzi2015}
G.-E. Espinosa and N.~Touzi.
\newblock Optimal investment under relative performance concerns.
\newblock \emph{Mathematical Finance}, 25\penalty0 (2):\penalty0 221--257,
  2015.

\bibitem[Fouque and Ichiba(2013)]{Fouque-Ichiba}
J.-P. Fouque and T.~Ichiba.
\newblock Stability in a model of interbank lending.
\newblock \emph{SIAM Journal on Financial Mathematics}, 4:\penalty0 784--803,
  2013.

\bibitem[Fouque and Sun(2013)]{Fouque-Sun}
J.-P. Fouque and L.-H. Sun.
\newblock Systemic risk illustrated.
\newblock \emph{Handbook on Systemic Risk, Eds J.-P. Fouque and J. Langsam},
  2013.

\bibitem[Fouque and Zhang(2018)]{FouqueZhang2018}
J.-P. Fouque and Z.~Zhang.
\newblock Mean field game with delay: A toy model.
\newblock \emph{Risks}, 6\penalty0 (3):\penalty0 90, 2018.

\bibitem[Garnier et~al.(2013{\natexlab{a}})Garnier, Papanicolaou, and
  Yang]{Garnier-Mean-Field}
J.~Garnier, G.~Papanicolaou, and T.-W. Yang.
\newblock Diversification in financial networks may increase systemic risk.
\newblock \emph{Handbook on Systemic Risk, Eds J.-P. Fouque and J. Langsam},
  2013{\natexlab{a}}.

\bibitem[Garnier et~al.(2013{\natexlab{b}})Garnier, Papanicolaou, and
  Yang]{GarnierPapanicolaouYang}
J.~Garnier, G.~Papanicolaou, and T.-W. Yang.
\newblock Large deviations for a mean field model of systemic risk.
\newblock \emph{SIAM Journal on Financial Mathematics}, 4:\penalty0 151--184,
  2013{\natexlab{b}}.

\bibitem[Garnier et~al.(2017)Garnier, Papanicolaou, and
  Yang]{Papanicolaou-Yang2015}
J.~Garnier, G.~Papanicolaou, and T.-W. Yang.
\newblock A risk analysis for a system stabilized by a central agent.
\newblock \emph{Risk and Decision Analysis}, 6\penalty0 (2):\penalty0 97--120,
  Aug 2017.

\bibitem[Huang et~al.(2006)Huang, Caines, and
  Malham{\'{e}}]{HuangCainesMalhame1}
M.~Huang, P.~E. Caines, and R.~P. Malham{\'{e}}.
\newblock Large population stochastic dynamic games: closed-loop
  {M}c{K}ean-{V}lasov systems and the {N}ash certainty equivalence principle.
\newblock \emph{Communications in Information and Systems}, 6:\penalty0
  221--252, 2006.

\bibitem[Huang et~al.(2007)Huang, Caines, and
  Malham{\'{e}}]{HuangCainesMalhame2}
M.~Huang, P.~E. Caines, and R.~P. Malham{\'{e}}.
\newblock Large population cost coupled {LQG} problems with nonuniform agents:
  individual mass behavior and decentralized {$\epsilon$}-{N}ash equilibria.
\newblock \emph{IEEE Transactions on Automatic Control}, 52:\penalty0
  1560--1571, 2007.

\bibitem[Karatzas and Shreve(1998)]{Karatzas2000}
I.~Karatzas and S.~E. Shreve.
\newblock \emph{Brownian Motion and Stochastic Calculus Second Edition}.
\newblock Springer-Verlag New York, 1998.

\bibitem[Lacker(2016)]{Lacker2015}
D.~Lacker.
\newblock A general characterization of the mean field limit for stochastic
  differential games.
\newblock \emph{Probability Theory and Related Fields}, 165\penalty0
  (3):\penalty0 581--648, 2016.

\bibitem[Lacker(2018)]{Lacker2018}
D.~Lacker.
\newblock On the convergence of closed-loop {N}ash equilibria to the mean field
  game limit.
\newblock \emph{arXiv:1808.02745}, 2018.

\bibitem[Lacker and Webster(2015)]{Lacker-Webster2014}
D.~Lacker and K.~Webster.
\newblock Translation invariant mean field games with common noise.
\newblock \emph{Electronic Communications in Probability}, 20\penalty0
  (42):\penalty0 1--13, 2015.

\bibitem[Lacker and Zariphopoulou(2019)]{Lacker-Zari2017}
D.~Lacker and T.~Zariphopoulou.
\newblock Mean field and n-agent games for optimal investment under relative
  performance criteria.
\newblock \emph{to appear in Mathematical Finance}, 2019.

\bibitem[Lasry and Lions(2006{\natexlab{a}})]{MFG1}
J.-M. Lasry and P.-L. Lions.
\newblock Jeux {\`{a}} champ moyen i. le cas stationnaire.
\newblock \emph{Comptes Rendus de l'Acad{\'{e}}mie des Sciences de Paris, ser.
  A}, 343\penalty0 (9), 2006{\natexlab{a}}.

\bibitem[Lasry and Lions(2006{\natexlab{b}})]{MFG2}
J.-M. Lasry and P.-L. Lions.
\newblock Jeux {\`{a}} champ moyen ii. horizon fini et contr{\^{o}}le optimal.
\newblock \emph{Comptes Rendus de l'Acad{\'{e}}mie des Sciences de Paris, ser.
  A}, 343\penalty0 (10), 2006{\natexlab{b}}.

\bibitem[Lasry and Lions(2007)]{MFG3}
J.-M. Lasry and P.-L. Lions.
\newblock Mean field games.
\newblock \emph{Japanese Journal of Mathematics}, 2\penalty0 (1):\penalty0
  229--260, Mar. 2007.

\bibitem[Sun(2017)]{Sun2016}
L.-H. Sun.
\newblock Systemic risk and interbank lending.
\newblock \emph{Journal of Optimization Theory and Applications}, 179\penalty0
  (2):\penalty0 400--424, 2017.

\end{thebibliography}
 \end{document}